\documentclass[11pt]{article}
\usepackage{biblatex}
\newcounter{savenumi}

\newtheorem{theoremfoo}{Theorem}
\newenvironment{theorem}{\pagebreak[1]\begin{theoremfoo}}{\end{theoremfoo}}

\newtheorem{propositionfoo}[theoremfoo]{Proposition}

\newtheorem{lemmafoo}[theoremfoo]{Lemma}
\newenvironment{lemma}{\pagebreak[1]\begin{lemmafoo}}{\end{lemmafoo}}
\newtheorem{claimfoo}[theoremfoo]{Claim}

\newtheorem{observationfoo}[theoremfoo]{Obervation}

\newtheorem{conjecturefoo}[theoremfoo]{Conjecture}

\newtheorem{corollaryfoo}[theoremfoo]{Corollary}
\newenvironment{corollary}{\pagebreak[1]\begin{corollaryfoo}}{\end{corollaryfoo}}
\newtheorem{exercisefoo}{Exercise}

\newtheorem{openfoo}[theoremfoo]{Question}

\newtheorem{nttn}[theoremfoo]{Notation}

\newtheorem{dfntn}[theoremfoo]{Definition}
\newenvironment{definition}{\pagebreak[1]\begin{dfntn}\rm}{\end{dfntn}}

\newenvironment{proof}
    {\pagebreak[1]{\narrower\noindent {\bf Proof:\quad\nopagebreak}}}{\QED}

\newcommand{\floor}[1]{\left\lfloor#1\right\rfloor}
\newcommand{\ceiling}[1]{\left\lceil#1\right\rceil}
\newcommand{\NP}{{\rm NP}}

\renewcommand{\P}{{\rm P}}      

\def\nre.{$n$\/-r.e.}
\newcommand{\scrod}{\quad\nopagebreak}

\newtheorem{factfoo}[theoremfoo]{Fact}
\newenvironment{fact}{\pagebreak[1]\begin{factfoo}}{\end{factfoo}}

\newenvironment{algorithm}{\begin{enumerate}}{\end{enumerate}}

\newtheorem{propertyfoo}[theoremfoo]{Property}

\makeatletter

\def\@makechapterhead#1{ \vspace*{50pt} { \parindent 0pt \raggedright 
 \ifnum \c@secnumdepth >\m@ne \huge\bf \@chapapp{} \thechapter. \par 
 \vskip 20pt \fi \Huge \bf #1\par 
 \nobreak \vskip 40pt } }

\def\@sect#1#2#3#4#5#6[#7]#8{\ifnum #2>\c@secnumdepth
     \def\@svsec{}\else 
     \refstepcounter{#1}\edef\@svsec{\csname the#1\endcsname.\hskip 1em }\fi
     \@tempskipa #5\relax
      \ifdim \@tempskipa>\z@ 
        \begingroup #6\relax
          \@hangfrom{\hskip #3\relax\@svsec}{\interlinepenalty \@M #8\par}
        \endgroup
       \csname #1mark\endcsname{#7}\addcontentsline
         {toc}{#1}{\ifnum #2>\c@secnumdepth \else
                      \protect\numberline{\csname the#1\endcsname}\fi
                    #7}\else
        \def\@svsechd{#6\hskip #3\@svsec #8\csname #1mark\endcsname
                      {#7}\addcontentsline
                           {toc}{#1}{\ifnum #2>\c@secnumdepth \else
                             \protect\numberline{\csname the#1\endcsname}\fi
                       #7}}\fi
     \@xsect{#5}}

\def\@begintheorem#1#2{\it \trivlist \item[\hskip \labelsep{\bf #1\ #2.}]}

\def\@opargbegintheorem#1#2#3{\it \trivlist
      \item[\hskip \labelsep{\bf #1\ #2\ (#3).}]}

\makeatother

\newif\ifshortconferences
\shortconferencesfalse
\newif\ifmediumconferences
\mediumconferencesfalse

\def\ending#1{{\count1=#1\relax
\count2=\count1
\divide\count2 by 100
\multiply\count2 by 100
\advance\count1 by -\count2
\ifnum\count1=11
th%
\else \ifnum\count1=12
th%
\else \ifnum\count1=13
th%
\else 
\count2=\count1
\divide\count1 by 10
\multiply\count1 by 10
\advance\count2 by -\count1
\ifnum\count2=1
st%
\else \ifnum\count2=2
nd%
\else \ifnum\count2=3
rd%
\else th%
\fi\fi\fi\fi\fi\fi
}}

\def\Proceedingsofthe{\ifshortconferences Proc.\else\ifmediumconferences Proc.\else Proceedings of the\fi\fi}

\newcounter{confnum}

\def\conf#1#2{%
\setcounter{confnum}{#2}%
\addtocounter{confnum}{-\csname #1zero\endcsname}%
\ifnum\value{confnum}=1%
\expandafter\ifx\csname #1One\endcsname\relax%
\Proceedingsofthe\ \arabic{confnum}\ending{\value{confnum}}\ \csname #1name\endcsname%
\else \csname #1One\endcsname\fi%
\else%
\Proceedingsofthe\
\arabic{confnum}\ending{\value{confnum}}\ \csname #1name\endcsname\fi}

\def\qsym{\vrule width0.7ex height0.9em depth0ex}

\newif\ifqed\qedtrue

\def\noqed{\global\qedfalse}

\def\qed{\ifqed{\penalty1000\unskip\nobreak\hfil\penalty50
\hskip2em\hbox{}\nobreak\hfil\qsym
\parfillskip=0pt \finalhyphendemerits=0\par\medskip}\fi\global\qedtrue}

\makeatletter
\def\eqnqed{\noqed
	\def\@tempa{equation}
	\ifx\@tempa\@currenvir\def\@eqnnum{\qsym}%
	\addtocounter{equation}{-1}\else%
    \def\@@eqncr{\let\@tempa\relax
    \ifcase\@eqcnt \def\@tempa{& & &}\or \def\@tempa{& &}%
      \else \def\@tempa{&}\fi
     \@tempa {\def\@eqnnum{{\qsym}}\@eqnnum}
     \global\@eqnswtrue\global\@eqcnt\z@\cr}\fi}

\def\eqnlabel#1#2{\if@filesw {\let\thepage\relax%
   \def\protect{\noexpand\noexpand\noexpand}%
   \edef\@tempa{\write\@auxout{\string
      \newlabel{#2}{{{#1}}{\thepage}}}}%
   \expandafter}\@tempa%
   \if@nobreak \ifvmode\nobreak\fi\fi\fi%
	\def\@tempa{equation}
	\ifx\@tempa\@currenvir\def\theequation{{#1}}%
	\addtocounter{equation}{-1}\else%
    \def\@@eqncr{\let\@tempa\relax
    \ifcase\@eqcnt \def\@tempa{& & &}\or \def\@tempa{& &}%
      \else \def\@tempa{&}\fi
     \@tempa {\def\@eqnnum{{#1}}\@eqnnum}
     \global\@eqnswtrue\global\@eqcnt\z@\cr}\fi}

\makeatother

\def\QED{\qed}

\makeatother

\usepackage[margin=1in]{geometry}
\usepackage{amsmath}
\usepackage{amsfonts}
\usepackage{graphicx}
\usepackage{comment}
\usepackage{algorithm} 
\usepackage{algpseudocode} 
\usepackage{setspace}
\algdef{SE}[SUBALG]{Indent}{EndIndent}{}{\algorithmicend\ }%
\algtext*{Indent}
\algtext*{EndIndent}

\begin{document}

\date{}
\begin{titlepage}
\title{Sublinear Approximation Schemes for Scheduling  Precedence Graphs of Bounded Depth
}

\author{ Bin Fu$^1$, Yumei Huo$^2$, Hairong Zhao$^3$
 \\
  bin.fu@utrgv.edu, yumei.huo@csi.cuny.edu, hairong@purdue.edu \\
$^{1}$
Department of Computer Science, 
University of Texas Rio Grande Valley\\
$^{2}$
Department of Computer Science,
College of Staten Island, CUNY\\
$^{3}$
Department of Computer Science, 
Purdue University Northwest\\
}
\maketitle

\begin{abstract}
We study the classical scheduling problem on parallel machines 
where the precedence graph has the bounded depth $h$. 
Our goal is to minimize the maximum completion time. We focus on developing  approximation algorithms that use only sublinear space or sublinear time.  We  develop the first one-pass streaming approximation schemes using sublinear space when all jobs' processing times differ no more than a constant factor $c$ and the number of machines $m$ is at most $\tfrac {2n \epsilon}{3 h c }$. This is so far the best approximation  we can have in terms of $m$, since no polynomial time approximation better than $\tfrac{4}{3}$ exists when $m = \tfrac{n}{3}$ unless P=NP.  
The algorithms are then extended to the more general problem where the largest $\alpha n$ jobs have no more than $c$ factor difference. 
We also develop the first sublinear time algorithms for both problems.  
For the more general problem,  when $ m  \le    \tfrac {  \alpha  n  \epsilon}{20 c^2  \cdot h  } $, our algorithm is a randomized $(1+\epsilon)$-approximation scheme 
that runs in sublinear time. 
This work not only provides an algorithmic solution to the studied problem under big data 
environment, but also gives a methodological framework for designing
sublinear approximation algorithms for other scheduling problems.
\end{abstract}

\end{titlepage}
\section{Introduction}
Big data and cloud computing play a huge role nowadays in our digital society. Each day a  large amount of data is  generated and collected by a variety of programs and applications. These large sets of data, which are referred as ``big data'', are hard to peruse or query on a regular computer. On the other hand, 
cloud computing provides a platform for processing big data efficiently 
on the “cloud” where the “cloud” is usually a set of high-powered servers from one of many providers. The ``cloud'' can view and query large data sets much more quickly than a standard computer could. Big data and cloud computing together provide the solutions for the companies with big data but limited resources, a dilemma encountered by many companies in manufacturing and service industries.

Two decades ago, researchers  in the area of statistics, graph theory, etc. started to investigate the sublinear approximation algorithms that 
uses only sublinear space or sublinear time, namely sublinear space algorithms or sublinear time algorithms. With more and more data being generated and stored away in the data center,  and higher and higher dimension of computation being required and performed remotely on the ``cloud'' in various applications, sublinear algorithms become a new paradigm in computing to solve the problems under big data 
and cloud computing.  Unlike the traditional data model where all the data can be stored and retrieved locally and one can hope to get the exact answers, the goal of sublinear space and sublinear time algorithms in general, is to obtain reasonably good approximate answers without storing or scanning the entire input.

Sublinear space algorithms are also called streaming algorithms, which process the input where some or all of the data is not available for random access in the local computers but rather arrives as a sequence of items and can be examined in only a few passes (typically just one). 
Early research on streaming algorithms dealt with simple statistics
of the input data streams, such as the
median~\cite{MunroPaterson80}, the number of distinct
elements~\cite{FlajoletMartin85}, or frequency
moments~\cite{AlonMatiasSzegedy96}. Recently, many effective streaming algorithms have been designed for a range of problems in statistics, optimization, and graph algorithms (see surveys by Muthukrishnan \cite{Muthu05}  and McGregor \cite{mcgregor14}). 

Sublinear time algorithms target at giving good approximations after inspecting only a very small portion of the input. Usually, this is achieved by using randomization. Sublinear time algorithms have been derived for many computational problems, for example, checking polygon intersections ~\cite{ChazelleLiuMagen05}, approximating the average degree in a graph~\cite{Feige06,GoldreichRon05}, estimating the cost of a minimum spanning tree~\cite{ChazelleRubinfeldTrevisan05,CzumajErgun05,CzumajSohler04},
finding geometric separators~\cite{FuChen06}, 
and property testing~\cite{GoldreichGoldwasserRon96,GoldreichRon00},
etc. 
Developing  sublinear time algorithms not only speeds up problem solving process, but also reveals  some interesting properties of computation, especially the power of randomization.

This paper aims at designing both types of sublinear approximation algorithms for the classical parallel machine  scheduling problem subject to precedence constraints. We hope that our algorithms not only  provide algorithmic solutions for this  specific problem under  big data and cloud computing environment, but also provide  a framework and insight for solving other scheduling problems under  big data and cloud computing environment which are encountered by many  companies in the manufacturing and service industries.

Formally our problem is to schedule $n$ jobs on $m$ identical parallel machines where there are precedence constraints between jobs.
The jobs are all available for processing at time $0$ and labeled as $1, 2, \cdots, n$. Each job $j$, $1 \le j \le n$, has a processing time $p_j$. The jobs have precedence constraints, $\prec$, such that $i \prec j$ represents that job $j$ cannot start until job $i$  finishes. The jobs and their precedence constraints can be described by a directed acyclic graph (DAG), $G= (V, E)$, where $V$ is a set of vertices representing the jobs and $E$ is a set of directed arcs representing the precedence constraints among the jobs.   
We assume that there are no transitive edges in $G$.  If  there is a directed arc $\langle i, j\rangle$ in $E$, 
then we have the precedence constraint $i \prec j$, and we say that job $i$ is the immediate predecessor of job $j$ and $j$ is the immediate successor of job $i$. We consider non-preemptive schedules, i.e. a job cannot be interrupted once it is started. 
 Given a schedule $S$, let $C_j$ be the completion time of job $j$ in $S$, then the makespan of the schedule $S$ is $C_{max} = \max_{1 \le j \le n} C_j$. The goal is to find the minimum makespan.
Using the three field  notation,   the problem can be denoted as $P \mid prec \mid C_{max}$ when the number of machines $m$ is arbitrary,  and be denoted as  $P_m \mid prec \mid C_{max}$ when $m$ is fixed.

A lot of research has been done on this classical scheduling problem. 
For arbitrary precedence,  when $m = 2$ and jobs have unit processing time, i.e., $P_2 \mid prec, p_j =1  \mid C_{max}$,  Coffman and Graham \cite{cg72} gave an optimal polynomial time algorithm in 1972. In 1978, Lenstra and Kan \cite{lk78} showed when jobs have unit processing time, the problem with arbitrary precedence constraints and arbitrary $m$, $P \mid prec, p_j =1  \mid C_{max}$, is strongly NP-hard.  When the jobs' processing times are either 1 or 2,  Lenstra and Kan \cite{lk78} and Ullman \cite{u75} independently showed that the problem $P_2 \mid prec, p_j = 1, 2 \mid C_{max}$ is strongly NP hard.
However,  the complexity of the problem  $P_3 \mid prec, p_j =1  \mid C_{max}$ remains open. 
Graham \cite{graham66} showed   that list scheduling is a $(2-\tfrac{1}{m})$-approximation for the problem with arbitrary $m$ and arbitrary job processing times, i.e., $P \mid prec \mid C_{max}$. In 2011, Svensson \cite{s11} showed that assuming a new, possibly stronger, version of the unique games conjecture (introduced by Bansal and Khot \cite{bk09}), it is NP-hard to approximate the scheduling problem, $P \mid prec, p_j=1 \mid C_{max}$, within any factor strictly less than 2. This result improves the inapproximibility of $4/3$ by Lenstra and Rinnooy Kan \cite{lk78}.

Due to the importance  and the hardness of the problem, a lot of research focused on various types of precedence constraints.  
One type of precedence constraints studied in literature is precedence graphs with bounded height where the height is the number of vertices on the longest path. 
In 1978, Lenstra and Kan \cite{lk78} showed that for arbitrary number of machines $m$, the problem is NP-hard even if the precedence constrained graph has bounded height and the jobs have unit processing time. In 1984, Doleva and Warmuth  \cite{dw84} developed an optimal algorithm for this problem when $m$ is fixed and the running time of the algorithm is $n^{h(m-1)+1}$. In 2006, Aho and M{\"{a}}kinen \cite{am06} considered a special case where both the height of the graph and the maximum degree are bounded, and jobs have unit processing time. They showed that for large $n$, the optimal schedule has makespan $\ceiling{n/m}$ and can be scheduled using modified critical path rule. This result is in fact a special case of the one studied by Dolev and Warmuth \cite{dw84}. For more related results, one can refer to the survey by Prot and Bellenguez-Morineaua   \cite{pb18} on how the structure of precedence constraints may change the complexity of the scheduling problems.

In this paper, we focus on the problem where the precedence  graph has bounded depth $h$. The depth of a job $j$, denoted as $dp_j$, is the number of jobs on the longest directed path ending at $j$ in $G$. It is easy to see that the maximum depth of the jobs in $G$ is equal to the height of the graph $G$. Given a precedence graph, 
one can easily compute  the depth $dp_j$ of each job $j$. 
In addition, we assume that the processing times of the jobs are constrained.   
We first consider the case  
that the processing times of the jobs vary from one to another within a  factor $c$, i.e.
$p_{max} \le c \cdot p_{min}$ where $p_{max} = max_{1 \le j \le n} \{p_j\}$, $p_{min} =  min_{1 \le j \le n} \{p_j\}$, and $c$ is a constant integer.  
Using the three-field  notation, we denote this problem as $P \mid prec, dp_j \le h, p_{max} \le c \cdot p_{min} \mid C_{max}$.  
We then  
consider the  more general case where the largest $\alpha n$ jobs have no more than $c$ factor difference for some constant where $0 < \alpha \le 1$. For a given set of $n$ jobs, let $[j]$ be the $j$-th smallest job. Then $p_{[1]}$ is the smallest job, and $p_{[n]}$ is the largest job. We  
denote this more generalized problem  as $P \mid prec, dp_j \le h,  p_{[n]} \le c \cdot p_{[ (1 - \alpha) n)]}  \mid C_{max}$. 
 Our goal is to develop sublinear approximation algorithms, that is, approximation algorithms using only sublinear time or sublinear space, for these two versions of precedence constrained scheduling problems.

\subsection{New Contributions}

In this work, we develop two types of sublinear approximation algorithms 
for the classical parallel machine  scheduling problems where the precedence graph has bounded depth $h$ and the processing times of jobs are constrained.  
Specifically, our   contributions are listed as follows:

\begin{enumerate}
 \item We develop two streaming approximation schemes  for the problem  $P  \mid prec, dp_j \le h, p_{max} \le c \cdot p_{min} \mid C_{max}$ 
depending on whether $c$, $h$, and each job's depth are known or not. 
The algorithms are then extended to solve the more general problem where the largest $\alpha n$ jobs have no more than $c$ factor difference for some constant $\alpha$, $0 < \alpha \le 1$, 
$P \mid prec, dp_j \le h,  p_{[n]} \le c \cdot p_{[ (1 - \alpha) n)]}  \mid C_{max}$. 
 
\item  
We develop the first randomized approximation  schemes that uses only sublinear time for both problems. In particular, 
for the more general problem, $P \mid prec, dp_j \le h,  p_{[n]} \le c \cdot p_{[ (1 - \alpha) n)]}  \mid C_{max}$,  when $ m  \le    \tfrac {  \alpha  n  \epsilon}{20 c^2  \cdot h  } $, our algorithm is a randomized $(1+\epsilon)$-approximation scheme 
that runs in time  $O( \tfrac{c^4 h^2 m^2} {\alpha^3 \epsilon^6} \log ^ 2 (\tfrac{c n}{\epsilon})  \log (\tfrac{h}{\epsilon} \log (\tfrac{c n}{\epsilon}  )))$.
\item  Our  approximation results greatly complement the in-approximability results of the studied problems.
     When $m = \tfrac{n}{3}$, even if $h=3$ and $c = 1$,  
     the problems cannot be approximated within a factor of $\tfrac{4}{3}$ in polynomial time   unless P=NP (see Section 2 for reference). 
     Surprisingly, our results show  that when $m$ is a little bit smaller, i.e.,   upper bounded by $n$ times a factor that depends on $\epsilon$, $h$ and $\alpha$, then the problems admit   polynomial time approximation schemes. For example, if   $m \le \tfrac{n}{15}$,  $h=3$ and $c = 1$, then there is a polynomial time $1.3$-approximation for $P \mid prec, dp_j \le h, p_{max} \le c \cdot p_{min} \mid C_{max}$.    
\item 
We provide   
a methodological framework for designing  sublinear approximation algorithms that can be used for solving other scheduling problems. 
The framework  
starts with generating the ``sketch of input'', which is a summarized description of the input jobs, then computes an approximate value of the optimal criterion, and finally generates the ``sketch of schedule'',
 a succinct description of a schedule that achieves the approximate value.

We introduce the concept of ``sketch of schedule'' for the applications where not only an approximate value, but also a schedule associated with the approximate value is needed. As illustrated in the paper, we can use the ``sketch of schedule'' to easily generate  a real schedule when the complete jobs information is  
read.
\end{enumerate}

The paper is organized as follows. In Section 2, we   give the complexity of the studied scheduling problems. In Section 3, we present the streaming algorithms for our problems. In Section 4, we design the randomized sublinear time algorithms for our problems. Finally, we draw the concluding remarks in Section 5.

\section{Complexity}
From the introduction, 
we know that if the jobs have unit processing time, then  $P_m \mid prec,  dp_j \le h, p_j = 1 \mid  C_{max}$ is solvable in $O(n^{h(m-1) +1})$ time which is polynomial if $m$ is constant (see \cite{dw84} for reference); however, the problem with arbitrary $m$, $P \mid prec,  dp_j \le h, p_j=1 \mid C_{max}$, is NP-hard in the strong sense even if $h = 3$ (see \cite{lk78} for reference). In this section, we first show that if we allow jobs to have different processing times,  then even for fixed $m$, the problem becomes   NP-hard.

\begin{theorem}
The problem $P_m \mid prec, dp_j \le h, p_{max} \le c \cdot p_{min} \mid C_{max}$ is $\NP$-hard.
\end{theorem}

\begin{proof}
We will reduce even-odd partition problem to a restricted even-odd partition problem, and then reduce the restricted even-odd partition problem to $P_2 \mid  p_{max} \le c \cdot p_{min} \mid  C_{max}$,  which implies that $P_m \mid  prec, dp_j \le h, p_{max} \le c \cdot p_{min} \mid  C_{max}$ is NP-hard.

{\bf Even-odd partition:}  
 there is a set of  $2n$ integers $B = \{b_i, 1 \le i \le 2n\}$ such that $b_{i} < b_{i+1}$ for all $1 \le i < 2n$, is there a partition of B into $B_1$ and $B_2$ such that  $B_1$ and hence $B_2$ contains  exactly one of $\{b_{2i-1}, b_{2i}\}$  for each $1 \le i \le n$, and $\sum_{b_i \in B_1} b_i  = \sum_{b_i \in B_2} b_i$?

{\bf Restricted even-odd partition:}  
Given a set of  $2n$ integers $D = \{d_i, 1 \le i \le 2n\}$ such that $d_{i} < d_{i+1}$ for all $1 \le i < 2n$, and $d_{2n} \le c d_1 $ for some constant $c > 1$, is there a partition of D into $D_1$ and $D_2$ such that    $D_1$ and hence $D_2$ contains exactly one of $\{d_{2i-1}, d_{2i}\}$ for each $1 \le i \le n$, and $\sum_{d_i \in D_1} d_i  = \sum_{d_i \in D_2} d_i$?

Given an arbitrary instance  $B = \{b_i, 1 \le i \le 2n\}$ of the even-odd partition problem, we can reduce it to an instance of  the restricted even-odd partition problem $D = \{d_i, 1 \le i \le 2n\}$ as follows. Without loss of generality, we can assume that $b_{2n} > c\cdot  b_1$.
Let $Y$ be the integer such that $Y \ge \tfrac{b_{2n} - c b_1 }{c-1}$, i.e. $b_{2n} \le    c \cdot b_1 +(c-1) Y  $ . For each $1 \le i \le 2n$, let $d_i = b_i + Y$. It is easy to see that $d_{2n} = b_{2n}+Y \le c  b_1 +  c \cdot Y    = c \cdot d_1$.
It is trivial to show that there is a solution to instance $B$ if and only if there is a solution for instance $D$. Thus the restricted even-odd partition is also NP-hard. 
The restricted even-odd partition can be easily reduced to the scheduling problem $P_2\mid  p_{max} \le c \cdot p_{min} \mid C_{max}$, which implies that $P_m\mid prec, dp_j \le h, p_{max} \le c \cdot  p_{min} \mid C_{max}$ is NP-hard. 
\end{proof}

The next theorem is showing the in-approximability of our problems. In the strong NP-hardness proof of 
$P\mid prec, p_j = 1 \mid C_{max}$ in \cite{lk78},  the scheduling instance  created from  the instance of clique problem  has a precedence graph of height $3$, and there is a schedule of the  $n=3m$ jobs with makespan of 3 if and only if there is a solution to the clique instance.   This implies if an approximation algorithm can generate a schedule with approximation ratio  less than $4/3$, it must be optimal, which is impossible unless P=NP.

\begin{theorem} \label{thm:hard-to-approximate} Given any $\epsilon>0$, unless \P=\NP,
there is no polynomial time $(4/3-\epsilon)$-approximation algorithm
for $P \mid prec, dp_j \le h, p_j = 1 \mid  C_{max}$ even if $h = 3$.
\end{theorem}

Despite the in-approximability result from Theorem~\ref{thm:hard-to-approximate}, in the next two sections, we will develop  approximation schemes that use only sublinear space or sublinear time for our studied problems when $m$ is upper bounded by $n$ times a factor. 

\section{Streaming Algorithms using Sublinear Space }\label{sec:stremaing-alg}

At the conceptual level, our streaming algorithms have the following two stages:
\begin{itemize}
    \item[] \textbf{Stage 1:} Generate and store a sketch of the input while reading the input stream.  
    \item[] \textbf{Stage 2:} Compute an approximation of the optimal value based on the sketch of the input.  
\end{itemize}

Roughly speaking, the sketch of the input is a summary of the input jobs which requires only sublinear space. Instead of storing the accurate processing times of the jobs, we map each job's processing time into the range of $[(1+\delta)^u,(1+\delta)^{u+1})$ where $\delta$ is a parameter. Thus we only need to store the number of jobs that are mapped in each range for each depth. 
We then use the rounded processing time  
for each job to obtain the approximation of the optimal makespan. 
A formal definition of the sketch for our problems is given below.  

\begin{definition}
For a given parameter $\delta$, and an instance of the problem $P \mid prec,  dp_j \le h, p_{max} \le c \cdot p_{min} \mid C_{max}$ or $P \mid prec, dp_j \le h, p_{[n]} \le c \cdot p_{[ (1 - \alpha) n)]} \mid C_{max}$, the \textbf{sketch of the input} with respect to  $\delta$,  
denoted as $SKJ_\delta = \{ (d,u, n_{d,u})\}$,  consists of a set of tuples, $(d, u, n_{d,u})$,  
where $n_{d,u}$ is the number of jobs with the depth $d$ and the processing time in the range of $[(1+\delta)^u,(1+\delta)^{u+1})$.  
\end{definition}

The size of the sketch, which is the number of tuples $(d,u, u_{d,u})$, may be different for different problems and different types of stream input.  In some cases, for example, we  disregard the jobs with small processing time.

In the following subsection, we will first present our streaming algorithms for the problem $P \mid prec,  dp_j \le h, p_{max} \le c \cdot p_{min} \mid C_{max}$. For the stream input, we consider both the case where $c$, $h$ and $dp_j$, $1 \le j \le n$, are given and the case  where these information is not directly given. We will then  adapt our algorithms to the more general problem $P \mid prec,  dp_j \le h,  p_{[n]} \le c \cdot p_{[ (1 - \alpha) n)]}  \mid C_{max}$.

\subsection{Streaming Approximation Schemes  for $P \mid prec,  dp_j \le h, p_{max} \le c \cdot p_{min} \mid C_{max}$}

\subsubsection{The parameters $c$, $h$, and $dp_j$ are known}
We study the problem under data stream model assuming $c$, $h$, and $dp_j$ are known. The jobs are given  via stream and each job $j$ is described by a pair $(p_j, dp_j)$, where $p_j$ and $dp_j$ are job $j$'s processing time and depth, respectively.  Without loss of generality, we can assume $p_j \in [1, c]$ in this case.
The algorithm is simple: 
scan the jobs from the stream 
and generate the sketch of the input, $SKJ_\delta = \{ (d,u, n_{d,u})\}$,  where $n_{d,u}$ is the number of jobs with the depth $d$ and the processing time in the range of $[(1+\delta)^u,(1+\delta)^{u+1})$; for each $d$, $1\le d \le h$, compute the length of the time interval where all the jobs with the depth $d$ can be feasibly scheduled, and then return the total length of these intervals. 
The complete algorithm is given in Streaming-Algorithm1. 

\begin{algorithm}
\renewcommand{\thealgorithm}{}	\caption {Streaming-Algorithm1}
   Input:   Parameters $\epsilon$, $m$, $c$ and $h$
         
  \hspace{0.4in}       Stream input: $(p_j, dp_j)$, $1 \le j \le n$.
         
    Output: An approximate value of the optimal makespan
\begin{algorithmic}[1]
   \State let $\delta = \tfrac{\epsilon}{3}$, $k = \floor{ \log_{1+\delta} c  }$
    \State read the input stream and generate the sketch of the input $SKJ_\delta$:
      \Indent
          \State  initialize the input sketch: $SKJ_\delta = \{ (d,u, n_{d,u}):  1 \le d \le h, 0 \le u \le k, n_{d,u} = 0\}$ 
      \For{each job $j$ with $(p_j, dp_j)$ in the stream input} 
        \State update $n_{d, u} = n_{d, u}+1$ where $d=dp_j$ and $u = \floor{\log_{1+ \delta} p_j}$
        \EndFor  
   \EndIndent
    \State compute the approximate makespan
       \Indent
        \State let $rp_k = c$
        \State for  each $u$, $0 \le u < k $
        \Indent
        \State let $rp_{u} =  (1 + \delta)^{u+1} $
        \EndIndent
       \For{ each $d$, $1 \le d \le h$}         
           \State let $A_d =   {\frac{1}{m} \sum_{u=0}^k(  n_{d,u}  \cdot rp_{u})}$ 
        \EndFor
        \State $A = \sum_{d=1}^h  (\floor{A_d} + c)$
        \EndIndent
        \State return $A$  
\end{algorithmic}
    \end{algorithm}

\begin{theorem}\label{thm-alg1}
For any $\epsilon$, when $ m   \le \tfrac  {2 n \epsilon} { 3 \cdot h \cdot c}$, Streaming-Algorithm1 is a one-pass streaming approximation scheme for $P  \mid prec, dp_j \le h, p_{max} \le c \cdot p_{min}   \mid C_{max}$ that uses $O({\tfrac{h \log c}{\epsilon}})$ space,   $O(1)$ update time for each job in the stream, and  $O({\tfrac{h \log c}{\epsilon}})$ time to return the approximate makespan.   
\end{theorem}
\begin{proof} 
First we consider the complexity. The space complexity is dominated by the sketch for which  we can use a two dimensional array of size   $ h \cdot k= O( \tfrac{ h \log c} { \epsilon })$. It is easy to see that the update time for each job is $O(1)$. Finally, it takes $O(h \cdot  k) = O({\tfrac{h \log c}{\epsilon}})$ time to compute and return the approximate value $A$. 

Now we consider the approximation ratio of the  algorithm. 
Let $I$ be the input instance, and let  
$I'$ be the instance corresponding to the sketch $SKJ_\delta$ which consists of $n_{d,u}$ jobs that have processing time $rp_{u}$  for each $d$, $u$.  Alternatively, we can also view $I'$ being obtained from $I$ by rounding up the processing time of each job. 
Let $C_{max}^*$  and $C'_{max}$ be the optimal makespan for the  instance $I$ and $I'$, respectively. It is easy to see that $ C_{max}^* \le C'_{max} \le (1 + \delta)  C_{max}^*$.  In the following, we prove that the returned value from Streaming-Algorithm1, $A$, satisfies the inequality, $C_{max}^* \le A \le (1 + \epsilon) C_{max}^*$. 

First, we show that 
$A$ is an upper bound of $C'_{max}$. 
Consider a schedule $S$ for the jobs of instance $I'$ which  schedules the jobs as follows:    the jobs at each depth $d$  are scheduled using list scheduling rule (i.e.  schedule the jobs one by one in the given order to the machine that is available at the earliest time), and jobs with depth $d+1$ can start only after  all jobs with depth $d$ complete. It is easy to see that in $S$ the jobs at  a  depth $d$ are scheduled into an interval of length at most $ \lfloor  \frac{1}{m} \sum_{u=0}^k(  n_{d,u}  \cdot rp_{u}) \rfloor + c = \floor{A_d} + c$. Therefore, the makespan of the feasible schedule $S$ is at most $\sum_{d=1}^h (\floor{A_d} + c) = A$, which implies that $A \ge C'_{max}$, where $C'_{max}$ is the optimal makespan for the instance $I'$.

On the other hand, it is obvious  that
$\sum_{d=1}^{h} A_d =   \sum_{d =1}^{h} ({\frac{1}{m} \sum_{u=0}^k(  n_{d,u}  \cdot rp_{u})})$  is a lower bound of  $C'_{max}$, and $ C_{max}^* \le C'_{max} \le (1 + \delta)  C_{max}^*$.
Thus, $$A  =    \sum_{d=1}^h  \left(  \floor  {A_d} + c\right) \le  \left( \sum_{d=1}^h    {A_d} \right)  + h \cdot  c
\le   C'_{max} + h \cdot  c
\le    (1 + \delta)  C_{max}^* + h \cdot  c.$$
Since $C_{max}^* \ge \tfrac{n}{m}$,   when  $ m   \le \tfrac  { 2 n \epsilon} { 3 \cdot h \cdot c}$, we have $ h \cdot c \le \tfrac {2 \epsilon}{3}  \tfrac {n} { m }  \le \tfrac {2 \epsilon}{3}  C_{max}^*$. Therefore,
$$A  \le  (1 + \delta) C_{max}^* + \tfrac {2 \epsilon}{3}  C_{max}^* = (1 + \tfrac{\epsilon}{3}) C_{max}^* + \tfrac {2 \epsilon}{3}  C_{max}^* \le (1 + \epsilon) C_{max}^* .$$

In summary, we have $ C_{max}^*  \le C'_{max} <  A \le (1 + \epsilon) C_{max}^*, $ and this completes the proof.
\end{proof}

Recall the inapproximability result of the problem $P \mid prec,   dp_j \le h, p_j = 1 \mid C_{max}$ from Theorem~\ref{thm:hard-to-approximate}, which tells us no approximation better than $\tfrac{4}{3}$ is possible in polynomial time  even if the height is bounded and all jobs have unit processing time unless P=NP. Our result from Theorem~\ref{thm-alg1} surprisingly shows that if $m$ is bounded by a fraction of $n$ then we can get a $(1+\epsilon)$-approximation even if the processing times are slightly different. 
For example, if $m \le \tfrac{n}{15}$,   $h=3$ and $c=1$, then we can get a $1.3$-approximation. 
 
Theorem~\ref{thm-alg1} also shows that Stream-Algorithm1 only takes constant time to read and process each job in the stream input, and then constant time and constant space to return an approximation of the makespan if the parameters are known. In some cases, however, we may not know the exact value of $c$, but we are given an estimate $\hat{c}$ of the parameter $c$. In these cases, we can  still apply Stream-Algorithm1 by using $\hat{c}$. As long as  $\tfrac{\hat{c}}{c}$ is a constant, we still have a $(1+\epsilon)$-approximation with the same space and time complexity.

\subsubsection{The parameters $c$, $h$, and $dp_j$, are unknown}
In this subsection, we consider the case that the parameters $c$ (or the estimate $\hat{c}$) and $h$ are not known, furthermore,  the depth of the jobs are not given directly as in the previous section. Instead, the stream input consists of all the jobs $(j, p_j)$ in arbitrary order followed by all the arcs $\langle i,j\rangle$ of the precedence graph in topological order.

In this case, we need to compute and update both the depth of each job and the sketch of the input dynamically as we read the input.  We use a B-tree to maintain the sketch tuples of the input $(d, u, n_{d,u})$ where $(d,u)$ is the key.  
We define a linear order to compare two tuples $(d_1, u_1, n_{d_1, u_1})$ and $(d_2, u_2, n_{d_2,u_2})$,  we say $(d_1, u_1, n_{d_1, u_1}) < (d_2, u_2, n_{d_2,u_2})$ if 1) $u_1 < u_2$, or 2)  $u_1 = u_2$ and $d_1 < d_2$.
Additionally we use an array $B$ to store the jobs' information: for each job $j$ with the processing time $p_j$, we maintain a pair $(dp_j, u_j)$, where $dp_j$ represents  its current depth, and $u_j = \floor{\log_{1+\delta} p_j}$.

When each job $(j, p_j)$ arrives in the stream input, we update job $j$'s entry in the array $B$ such that  $dp_j=1$ and $u_j = \floor{\log_{1+\delta} p_j}$, then create and insert a node $(1,u_j, n_{1,u_j})$ into the tree. Simultaneously we update the smallest processing time $p_{min}$ and the largest processing time $p_{max}$. After all the jobs are read in, we can get the final $p_{min}$ and $p_{max}$ and compute $c= \lceil{p_{max}}/{p_{min}} \rceil$. 

When each arc $\langle i, j\rangle$, which indicates job $i$ is the direct predecessor of job $j$, arrives in the stream input, we access job $i$'s and $j$'s entries in the array to obtain their keys $(d_i, u_i)$ and $(d_j, u_j)$, and compute $dp_j=\max{(d_j, d_i+1)}$. 
If $dp_j > d_j$, we will update the node $(d_j, u_j, n_{d_j, u_j})$ by setting $n_{d_j, u_j} = n_{d_j, u_j}-1$ or delete this node if $n_{d_j, u_j}$ becomes $0$; and then update the node $(dp_j, u_j, n_{dp_j, u_j})$ by setting $n_{dp_j, u_j} = n_{dp_ju_j}+1$ or insert a new node if the node with the key $(dp_j, u_j)$ does not exist in the tree. The job $j$'s entry in the array $B$ is also updated with $(dp_j, u_j)$. After all the arcs are read in, we can get the sketch of the stream input and the largest depth $h$. 
The complete algorithm is given in Streaming-Algorithm2.

\begin{algorithm}
\renewcommand{\thealgorithm}{}	\caption {Streaming-Algorithm2}
   Input:   Parameters $\epsilon$, $m$ 
         
  \hspace{0.4in}    Stream input: the set of jobs in arbitrary order, $(j, p_j)$, $1 \le j \le n$,  followed by 
  
 \hspace{0.4in} the set of arcs of the precedence graph in topological order 
    Output: An approximate value of the optimal makespan
\begin{algorithmic}[1]
    \State create an empty B-tree $T$ and an array $B$ of size $n$
    \State initialize $p_{min}=\infty$, $p_{max}=1$, $h=1$
    \State let $\delta = \tfrac{\epsilon}{3}$
    \State read the input stream and generate the sketch of the input $SKJ_\delta$:
      \Indent    \For{each job $j$ with $(p_j, dp_j)$ in the stream input} 
                      \State let $u = \floor{\log_{1+\delta} p_j}$  
            \State $B[j] = (1, u)$ 
              \If{  there is a node  $(1, u, n_{1,u})$ in the tree $T$}
                      \State    update this node by setting $n_{1,u}=n_{1,u}+1$
               \Else  
                    \State create and insert a node  $(1, u, 1)$ into $T$
                    \EndIf
             \State if $p_{min} > p_j$, $p_{min}=p_j$
            \State if $p_{max} < p_j$, $p_{max}=p_j$
        \EndFor  
  \For{ each arc $\langle i, j\rangle$,   in the stream input}
       
            \State  let  $(d_i, u_i) = B[i]$ and $(d_j, u_j) = B[j]$
            \If {$d_i+1  > d_j$  }       
                \State  $ dp_j =  d_i+1$ 
                \State  $B[j] = (dp_j, u_j)$
                \State update the node $(d_j, u_j, n_{d_j, u_j} )$ in $T$ as below
                 \Indent
                     \State  $n_{d_j, u_j} = n_{d_j, u_j}-1$
                      \State if $n_{d_j, u_j}=0$, delete this node 
                  \EndIndent
                \If{ the node with the key $(dp_j, u_j)$ does not exist in the tree}
                \State  insert a new node $(dp_j, u_j, 1)$
                \Else 
            \State update the node $(dp_j, u_j, n_{dp_j, u_j})$ in $T$ by setting $n_{dp_j, u_j} = n_{dp_j, u_j}+1$               
             \EndIf
          \EndIf
            \State if $h < dp_j $, set $h=dp_j$
       \EndFor
      \State traverse all the nodes $(d,u, n_{d,u})$ in $T$
         \State \hspace{0.2in}   let $SKJ_\delta = \{ (d,u, n_{d,u})\}$         
    \EndIndent
    \State compute the approximate makespan
       \Indent
         \State let $u_{-}=\floor{\log_{1+\delta} p_{min}}$ and $u_{+}=\floor{\log_{1+\delta} p_{max}}$
         \State let $rp_{u_{+}} = p_{max}$
         \State for each $ u_{-} \le  u < u_{+} $
         \Indent 
         \State let  $rp_{u} =  (1 + \delta)^{u+1} $
         \EndIndent
        \State  for each $d$
        \State \hspace{0.2in}let $A_d =  {\frac{1}{m} \sum_{u=u_{-}}^{u_{+}}(  n_{d,u}  \cdot rp_{u})}$   
        \State let $A = \sum_{d=1}^h  (\floor{A_d} + p_{max})$
          \EndIndent  
  \State  return $A$  
\end{algorithmic}
\end{algorithm}

\begin{theorem}\label{thm-alg2}
If parameters $c$ and $h$ are not known, both the jobs and the precedence graph in topological order are input via the stream, for any $\epsilon$, when   $ m   \le \tfrac  { 2 n \epsilon} { 3 \cdot h \cdot c}$, Streaming-Algorithm2 is a one-pass streaming approximation scheme for $P \mid prec,  dp_j \le h, p_{max} \le c \cdot p_{min} \mid C_{max}$ that uses $O(n)$ space, takes $O(  \log ( \tfrac{h}{\epsilon}  \log c ))$ update time for processing each job and each arc in the stream, 
  and  $O({\tfrac{h \log c}{\epsilon}})$ time to return the approximate makespan.  
\end{theorem}

\begin{proof} 
The main difference of Streaming-Algorithm2 and Streaming-Algorithm1 is the implementation. The analysis for approximation ratio is similar to Theorem~\ref{thm-alg1}. We will use the same notations as in the proof of Theorem~\ref{thm-alg1}. So $C^*_{max}$ is the optimal makespan of the input instance, $C'_{max}$ is the optimal makespan for the instance  $I'$ corresponding to the sketch $SKJ_{\delta}$.  We can construct a schedule $S$ for $I'$ whose makespan is at most 
$\sum_{d=1}^h (\floor{A_d} + p_{max}) = A$, which implies that $A \ge C'_{max}$.  

It is obvious  that 
$C'_{max} \ge \sum_{d=1}^{h} A_d $. 
Thus, $$A  =   \sum_{d=1}^h  \left(  \floor  {A_d} + p_{max}\right) \le  \left( \sum_{d=1}^h    {A_d} \right)  + h \cdot  p_{max}
\le   C'_{max} + h \cdot   p_{max}
\le    (1 + \delta)  C_{max}^* + h \cdot   p_{max}.$$
Since  $p_{max} \le c \cdot p_{min}$ and $C_{max}^* \ge \tfrac{n \cdot p_{min}}{m}$,  we have
$ h \cdot p_{max} \le  h \cdot c \cdot p_{min} \le h \cdot c \cdot \tfrac{m}{n} C_{max}^*$. when  $ m   \le \tfrac  { 2 n \epsilon} { 3 \cdot h \cdot c}$, we get $ h \cdot p_{max}   \le    \tfrac {2 \epsilon}{3}  C_{max}^*.$ 
Therefore,
$$A  \le  (1 + \delta) C_{max}^* + \tfrac {2 \epsilon}{3}  C_{max}^* = (1 + \tfrac{\epsilon}{3}) C_{max}^* + \tfrac {2 \epsilon}{3}  C_{max}^* =   (1 + \epsilon) C_{max}^* .$$
In summary, we have $ C_{max}^*  \le C'_{max} \le  A \le (1 + \epsilon) C_{max}^* $

Now we analyze the complexity of Streaming-Algorithm2, the number of nodes in B-tree $T$ is at most $O(h \log_{1+\delta} \lceil \tfrac{p_{max}}{p_{min}} \rceil) = O(h \log_{1+\delta} c) = O(\tfrac{h} {\epsilon} \log c)$. So when each job or arc is read from the stream input, the corresponding update time for search, insertion or update operation on the B-tree is always $O(\log (h \log_{1+\delta} c)) =  O(\log ( \tfrac{h}{\epsilon}   \log c))$. 
After the input is read in, it takes additional $O(\tfrac{h} {\epsilon} \log c)$ time to traverse B-tree  and compute the approximation of the optimal value.
The stream input size is $O(n+e)$, where $n$ is the number of jobs and $e$ is the number of arcs of the precedence graph.  
Streaming-Algorithm2, uses only   $O(n)$ space to store  the array $B$ and  the tree  $T$, which is sublinear considering the number of arcs usually has $e=O(n^{1+\beta})$, $0 < \beta  \le 1$ in a dense graph.
\end{proof}

\subsection{Streaming Approximation Algorithms for $P \mid prec,  dp_j \le h, p_{[n]} \le c \cdot p_{[ (1 - \alpha) n)]} \mid C_{max}$ }

In this section, we  consider the more general case where the largest $\alpha n$ jobs have no more than $c$ factor difference for some constant $0 < \alpha \le 1$. Apparently, the problem $P \mid prec,  dp_j \le h, p_{max} \le c \cdot p_{min} \mid C_{max}$ is  the special case where $\alpha = 1$. Following the same procedure of our streaming algorithms, 
we need to compute the sketch of the input $SKJ_\delta = \{(d, u, n_{d,u})\}$. 
However, different from the case $\alpha = 1$, i.e., the problem $P \mid prec,  dp_j \le h, p_{max} \le c \cdot p_{min} \mid C_{max}$,
for which there are only constant number $O(\tfrac{  h \log c}{\epsilon})$ of entries in the sketch of the input, for the problem $P \mid prec,  dp_j \le h, p_{[n]} \le c \cdot p_{[ (1 - \alpha) n)]} \mid C_{max}$, if we consider all jobs in the sketch there may be a very large number of entries in the sketch of the input since $p_{max}$ may be very large compared with $p_{min}$.  
We will show in the following that when we generate the sketch of the input we can ignore those small jobs whose  processing time is less than $\tfrac{p_{max}}{n^2}$ and still get a good approximation  of the optimal makespan using only sublinear space.

\subsubsection{ The parameters $c$, $h$, and $dp_j$, are known}
We study the streaming algorithm for $P \mid prec, dp_j \le h, p_{[n]} \le c \cdot p_{[ (1 - \alpha) n)]} \mid C_{max}$ when the parameters $c$, $h$, and $dp_j$ for all $1 \le j \le n$ are known. 
As mentioned above, the jobs with processing time less than $\tfrac{p_{max}}{n^2}$ will not be included in the sketch $SKJ_\delta$. 
Specifically, $SKJ_\delta = \{(d, u, n_{d,u}): u_{-} \le u \le u_{+}$, $1 \le d \le h  \}$, where $u_{-} = \lfloor \log_{1+\delta} \tfrac{p_{max}}{n^2} \rfloor$, and $u_{+} =\lfloor \log_{1+\delta} p_{max} \rfloor$.
Without loss of generality, we assume $p_{max}$ is not known until all input is read.  So $p_{max}$ in our algorithm represents  the current maximum processing time   of the jobs that  we have read so far.
We use a $B$-tree to store all the considered tuples, $(d, u, n_{d,u})$. When a job $j$ with $(p_j, dp_j)$  arrives, if $p_j < \tfrac{p_{max}}{n^2}$, we skip this job and continue to read the next job. Otherwise, let $d=dp_j$ and $u=\lfloor\log_{1+\delta} p_j \rfloor$, and we update B-tree as follows: if $(d, u, n_{d,u})$ exists in the tree, update this node with $(d, u, n_{d,u}+1)$; otherwise, insert a new node $(d, u, 1)$. To limit the number of nodes in the tree, whenever a new node is inserted,  we check   the node with  the smallest $u$, $(d', u', n_{d',u'})$, if $u' < \lfloor  \log_{1+\delta} \tfrac{p_{max}}{n^2} \rfloor$, we delete the smallest node. 
The final sketch of the input $SKJ_\delta$ includes only  the tuples $(d, u, n_{d,u})$  from the $B$-tree  such that $u_{-} \le u \le u_{+}$.
 We present our algorithm formally in Streaming-Algorithm3.

\begin{algorithm}
\renewcommand{\thealgorithm}{}	\caption {Streaming-Algorithm3}
   Input:   Parameters $\epsilon$, $m$, $n$, $\alpha$, $c$ and $h$
         
    \hspace{0.4in}     Stream input:  $(p_j, dp_j)$ for all jobs $1 \le j \le n$.
         
    Output: An approximate value of the optimal makespan
\begin{algorithmic}[1]
   \State let $\delta = \tfrac{\epsilon}{3}$
    \State create an empty B-tree $T$
    \State initialize $p_{max}=1$
    \State read the input stream and generate the sketch of the input $SKJ_\delta$:
     \Indent
        \For {each job $j$ with $(j, p_j)$ in the stream input}
       
        \If {$p_j < \tfrac{p_{max}}{n^2}$}
        \State skip this job and continue the next job
        \Else 
        
        \State if $p_{max} < p_j$, $p_{max}=p_j$
        \State let  $d=dp_j$, $u=\floor{ \log_{1+\delta} p_j}$, and update B-tree as follows:
           \If {node $(d, u, n_{d,u})$ exists in the tree }
             \State  \hspace{0.2in}  update  the node with  $  n_{d,u} = n_{d,u}+1$   
            \Else
                \State     insert a new node $(d, u, 1)$
                 \State   let $(d', u', n_{d',u'})$ be the node with the smallest $u$ 
                \State    if $  u'  <   \log_{1+\delta}\tfrac{p_{max}}{n^2}$, delete $(d', u', n_{d',u'})$  from the tree
            \EndIf
            \EndIf
        \EndFor    
        \State let $u_{-} = \floor{\log_{1+\delta} \tfrac{p_{max}}{n^2}}$,  $u_{+} =\floor{\log_{1+\delta} p_{max}}$
        \State traverse   $T$ and generate the sketch  using only the nodes with $u_{-} \le u \le u_{+}$ 
        \State \hspace*{0,2in} $SKJ_\delta =   \{ (d,u, n_{d,u} ): 1 \le d \le h, u_{-} \le u \le u_{+}\} $    
     \EndIndent
     \State compute the approximate makespan         
        \Indent 
         \State let $rp_{u_{+}} = p_{max}$
         \State  for each $ u_{-} \le  u < u_{+} $ 
         \Indent  \State let  $rp_{u} =  (1 + \delta)^{u+1} $
         \EndIndent
         \State  for each $d$ 
          \Indent
         \State let $A_d =      {\frac{1}{m} \sum_{u=u_{-}}^{u_{+}}(  n_{d,u}  \cdot rp_{u})}$ 
          \EndIndent
        \State let $A =( \sum_{d=1}^h  (\floor{A_d} + p_{max}))  + \lceil \tfrac{p_{max}}{n} \rceil$
       \EndIndent
      \State  return $A$  
\end{algorithmic}
\end{algorithm}

\begin{theorem}\label{thm-alg3}
When $ m  \le \tfrac { 2 n \alpha  \epsilon} { 3 (h+1) \cdot c} $, Streaming-Algorithm3 is a streaming approximation scheme  for the  problem $P \mid prec,  dp_j \le h,
 p_{[n]} \le c \cdot p_{[ (1 - \alpha) n)]},\mid C_{max}$ that uses $O(\tfrac{h}{\epsilon} \log n )$ space,  takes  $O(\log \tfrac{h}{\epsilon} + \log \log n)$ update time for each job in the stream, and $O(\tfrac{h}{\epsilon} \log n )$ time to return an approximate value that is at most $(1+\epsilon)$ times the optimal makespan.
\end{theorem}

\begin{proof}  
We first analyze the approximation ratio. 
Let $I$ be the given instance. Let  $I'$ be the  instance  obtained from $I$ by rounding up all the jobs with the processing times greater than or equal to $\tfrac{p_{max}}{n^2}$, 
i.e. for each job $j$ in $I$, if $p_j \ge \tfrac{p_{max}}{n^2}$, we round it up to 
$ rp_u $ where $u = \lfloor\log_{1 + \delta} p_j \rfloor$; otherwise, we keep it same as before. Let $C_{max}^*$ and $ C'_{max}$ be the optimal makespan for $I$ and $I'$ respectively. 
Let $I''$ be  the instance  corresponding to the sketch $SKJ_\delta$. Apparently $I''$ can be obtained from $I'$ by removing the small jobs whose processing time is less than $\tfrac{p_{max}}{n^2}$. Let $C''_{max}$ be the optimal makespan for $I''$. Then we have $C''_{max} \ge \sum_{d=1}^h A_d$. It is easy to see that $C''_{max} \le C'_{max} \le (1 + \delta)  C_{max}^*$, and  $C_{max}^* \le C'_{max} \le C''_{max} + n \cdot \tfrac{p_{max}}{n^2} = C''_{max} +  \tfrac{ p_{max}}{n}$. 

As before, we can construct a schedule $S$ for $I''$ using list scheduling rule to schedule the jobs depth by depth starting with $d =1$. To get a schedule for all jobs in $I'$ based on $S$, for each depth $d$, we can simply insert into $S$ all the small jobs of this depth onto the first machine   after all  big jobs  of depth $d$ finish and before the first big job of   $d+1$ starts.  Let the new schedule be $S'$. Apparently the makespan of $S'$  is at least $C_{max}^*$  and at most  $A=\sum_{d=1}^{h} (\floor{A_d} + p_{max}) + \lceil \tfrac{ p_{max}}{n} \rceil$. Thus, we have  $A \ge C_{max}^*$ 
 and 
\begin{equation} \label{eq:alg3-error}
A =  \left(\sum_{d=1}^h  (\floor{A_d} + p_{max}) \right) +\lceil \tfrac{ p_{max}}{n} \rceil \le C''_{max} + h \cdot  p_{max} +\lceil \tfrac{ p_{max}}{n} \rceil.\end{equation}
Since the largest $\alpha n$ jobs have no more than $c$ factor difference,   each of  the largest $\alpha n$ jobs has   processing time  at least 
$\tfrac{p_{max}} {c}$. Thus, we have $$C_{max}^* \ge \alpha \cdot n \cdot \tfrac{p_{max}} {c} \cdot \tfrac{1}{m} = \tfrac{ \alpha n}{    c \cdot m} p_{max},$$
which implies $p_{max} \le \tfrac{ c \cdot m}{\alpha \cdot n } C^*_{max}$. If we plug this into   inequality~(\ref{eq:alg3-error}), we get 
 \begin{eqnarray*}
 A  
 & \le & C''_{max} + h \cdot  p_{max} +\lceil \tfrac{ p_{max}}{n} \rceil\\
    &\le&  C''_{max} + (h+1) \cdot   p_{max} \\
 &\le& C''_{max}  + (h + 1)\cdot   \tfrac{  c \cdot m}{ \alpha \cdot n}  C^*_{max} \\
&\le&  (1 + \delta) C_{max}^* + \tfrac{(h + 1) \cdot  c \cdot m}{    \alpha n}C_{max}^*\\
&\le&  (1 + \delta + \tfrac{(h + 1) \cdot  c \cdot m}{    \alpha n})C_{max}^* \\
&\le&  (1 + \tfrac{\epsilon}{3} + \tfrac{(h + 1) \cdot  c \cdot m}{    \alpha n})C_{max}^*.
\end{eqnarray*}
If $ m  \le \tfrac { 2 n \alpha  \epsilon} { 3 (h+1) \cdot c} $, we have $ A  \le (1 + \tfrac{\epsilon}{3} + \tfrac{2\epsilon}{3}) C_{max}^* = (1 + \epsilon) C_{max}^* $.

Now we consider the complexity. The space complexity is dominated by the B-tree. As the way it is implemented, each time a node is inserted into the tree, if there is a node $(d,u, n_{d,u})$ such  that $u <  \log_{1+\delta} \tfrac{p_{max}}{n^2}  $, the smallest such node will be deleted from the tree. 
In this way, 
 the number of nodes in the tree is at most $ h \cdot \log_{1+\delta} n^2 = O(\tfrac{h}{\epsilon} \log n)$. For each job in the stream input, a constant number of tree operations are needed, and thus the update time for processing each job is $O(\log(\tfrac{h}{\epsilon} \log n) = O(\log (\tfrac{h}{\epsilon}) + \log \log n)$ time. 
After reading all the jobs, the computation of the approximation is bounded by the size of the sketch which is $O(\tfrac{h}{\epsilon} \log n)$.
\end{proof}

\subsubsection{The  parameters $c$, $h$, and $dp_j$ are unknown}
We consider the problem $P \mid prec,  dp_j \le h, p_{[n]} \le c \cdot p_{[ (1 - \alpha) n)]} \mid C_{max}$ when the parameters $c$, $h$, and $dp_j$ are not known. In this case, the stream input would include the jobs followed by the arcs.   As in Streaming-Algorithm 2,  we   use a  B-tree to store the sketch information, 
and   an array to store the jobs' information.
Both the array and the tree are updated when we read the jobs and arcs from the stream input.    
The streaming algorithm will be similar to Streaming-Algorithm2 but with some nodes for small processing times excluded  as in Streaming-Algorithm3.
Using similar arguments as in the proof of Theorem 
~\ref{thm-alg2} and \ref{thm-alg3}, we can get the following theorem.

\begin{theorem}\label{thm-alpha-n-unknow-parameter}
If parameters $c$ 
and $h$ are not known,  the jobs, and the precedence graph  in topological order are input via the stream, for any $\epsilon$, 
when $ m  \le \tfrac { 2  n \alpha \epsilon} { 3 (h+1) \cdot c} $, there is a streaming approximation scheme  for the  problem $P \mid prec, dp_j \le h, p_{[n]} \le c \cdot p_{[ (1 - \alpha) n)]} \mid C_{max}$ that uses $O(n )$ space,   takes  $O(\log \tfrac{h}{\epsilon} + \log \log n)$ update time for each job in the stream, and $O(\tfrac{h}{\epsilon} \log n )$ time to return the approximate value.
\end{theorem}

\subsection{The Sketch of the Schedule}\label{sec:Stream-Sketch-Schedule} 
All the   streaming algorithms we have presented so far   return an   approximate value of the optimal makespan. This may be  sufficient for some scheduling and planning applications.  
However, in many other applications,  
it would be desirable to  have a schedule whose makespan is the approximate value. In  the traditional data model, a schedule   explicitly or implicitly specifies for each job on which machine and in what time interval it is scheduled. This  means we need at least $\Omega(n)$ time complexity and space complexity to describe a schedule.  
For the big data model,   
we introduce the concept of {\bf sketch of a schedule} which is a condensed description of a schedule using only sublinear space. In the following, we first give a formal definition for the {\bf sketch of a schedule}, then we show that our previous algorithms can  compute simultaneously not only an approximate value, but also  the {\bf sketch of a schedule}, and finally we show how the sketch can be used to generate a real schedule that achieves the approximate value when the jobs are scanned in the second pass.

\begin{definition}
    For the problems $P \mid prec, dp_j \le h,  p_{max} \le c \cdot p_{max} \mid C_{max}$ and 
    $P \mid prec, dp_j \le h,  p_{[n]} \le c \cdot p_{[ (1 - \alpha) n)]} \mid C_{max}$, the {\bf sketch of a schedule} describes a feasible schedule 
    and consists of  a set of time instants $t_d$, $1 \le d \le h$, such that 
    all the jobs of depth $d$ can be scheduled during the interval $[t_{d-1}, t_d)$ for $1 \le d \le h$ where $t_0 = 0$. Mathematically we denote the {\bf sketch of a schedule} as $SKS = \{t_d: 1 \le d \le h\}$. 
\end{definition}

For the problem $P \mid prec, dp_j \le h, p_{max} \le c \cdot p_{min} \mid C_{max}$, 
by the proof of Theorem~\ref{thm-alg1} 
we know all the jobs of depth $d$ can be feasibly scheduled during an interval of length $\floor{A_d} + c $, which implies that  $SKS = \{t_d: 1 \le d \le h\}$ where $t_d = \sum_{k=1}^d (\floor{A_d} + c)$ for all $1 \le d \le h$, is the sketch of a schedule that can be computed by 
Streaming-Algorithm1. 
\begin{lemma}\label{lemma:sreaming-alg-sketch}
For the problem $P \mid prec, dp_j \le h, p_{max} \le c \cdot p_{min} \mid C_{max}$, Streaming-Algorithm1 
can compute a sketch of a schedule $SKS = \{t_d: 1 \le d \le h\}$ where $t_d = \sum_{k=1}^d (\floor{A_d} + c)$ for all $1 \le d \le h$.
\end{lemma}

Based on the sketch of the schedule, if we scan all the jobs in the second time, we can generate a feasible schedule using the Algorithm SketchToSchedule.

\begin{algorithm}
\renewcommand{\thealgorithm}{}
	\caption {SketchToSchedule}
   Input:   $SKS = \{t_d: 1 \le d \le h\}$ 
         
   \hspace{0.4in}        Stream input:  $(p_j, dp_j)$ for all jobs $1 \le j \le n$.
         
    Output: a feasible schedule $S$ of all the jobs
\begin{algorithmic}[1]
  \For {each $d$, $1 \le d \le h$}
     \State let $curM_d$ and $curT_d$ be the machine and time   where next job with depth d will be scheduled. 
     \State  $curM_d = 1$; $curT_d= t_{d-1}$ 
   \EndFor
  \State read the job stream input and generate the schedule:
   \Indent
        \For {each job $(p_j, dp_j)$  }     
         \State let $d = dp_j$
          \If { $curT_d + p_j \le t_{d}$  }
       \State    schedule job $j$ at time $curT_d$ on machine $curM_d$
         \State  set $curT_d = curT_d+p_j$
         \Else 
        \State schedule job $j$ at time $t_{d-1}$ on machine $curM_d+1$ 
        \State  set $curM_d  = curM_d + 1 $ and $curT_d  = t_{d-1}+p_j$
        \EndIf
        \EndFor
        \EndIndent
\end{algorithmic}
\end{algorithm}

By lemma~\ref{lemma:sreaming-alg-sketch} and the     Algorithm SketchToSchedule, we have the following theorem.

 \begin{theorem}\label{thm: sketch-alg1}
  For $P \mid prec, dp_j \le h,  p_{min} \le c \cdot p_{max}
 \mid C_{max}$, given   any $0 < \epsilon < 1$, when $ m   \le \tfrac  { 2 n \epsilon} { 3 \cdot h \cdot c}$, Streaming-Algorithm1 can compute  a sketch of the schedule $SKS$ which can be applied to Algorithm SketchToSchedule to generate  a feasible schedule with the makespan at most $(1+\epsilon)$ times the optimal makespan.
 \end{theorem}
 \begin{proof}
 Since the total length of the jobs at depth $d$ after rounding  is $m \cdot A_d$, and the largest processing time is $c$, 
it is easy to see that Algorithm SketchToSchedule generates a feasible schedule of these jobs 
 during the interval $[t_{d-1}, t_d]$, where $t_d = t_{d-1} + \floor{A_d} + c$. 
The final schedule of all $n$ jobs has the makespan at most $ t_h = \sum_{d=1}^h \floor{A_d} + c$, which is at most by $(1 + \epsilon)C_{max}^*$ by the proof of  Theorem~\ref{thm-alg1}.
\end{proof}

Similarly,  Streaming-Algorithm2  
can compute a sketch of a schedule $SKS = \{t_d: 1 \le d \le h\}$ where $t_d = \sum_{k=1}^d (\floor{A_d} + p_{max})$ for all $1 \le d \le h$, and we have 
the following theorem.
\begin{theorem}\label{thm: sketch-alg2}
  For $P \mid prec, dp_j \le h, p_{min} \le c \cdot p_{max} \mid C_{max}$, given   any $0 < \epsilon < 1$, when $ m   \le \tfrac  { 2 n \epsilon} { 3 \cdot h \cdot c}$  Streaming-Algorithm2 can compute  
  a sketch of a schedule  $SKS$ which can be applied to Algorithm SketchToSchedule to generate  a feasible schedule with the makespan at most $(1+\epsilon)$ times the optimal makespan. 
 \end{theorem}
 
For the problem $P \mid prec, dp_j \le h, p_{[n]} \le c \cdot p_{[ (1 - \alpha) n)]} \mid C_{max}$, Streaming-Algorithm3 gives an approximate value of the optimal makespan. However, the small jobs from depth $d$ are not considered when we calculate $A_d$, so the sketch of the schedule is slightly different from previous problem. We will show that  in this case, the sketch of a schedule is given by  $SKS = \{t_d: t_d = t_{d-1} + ( \floor{A_d} + p_{max} + \lceil \tfrac{p_{max}}{n}\rceil), 1 \le d \le h \}$, $t_0 = 0$.   

\begin{theorem}\label{thm: sketch-alg3}
  For the  problem $P \mid prec, dp_j \le h, p_{[n]} \le c \cdot p_{[ (1 - \alpha) n)]} \mid C_{max}$, given   any $0 < \epsilon < 1$, when $ m  \le \tfrac { 2 n \alpha  \epsilon} { 3 (h+1) \cdot c} $, 
Streaming-Algorithm3 can compute  a sketch of the schedule $SKS= \{t_d: t_d = t_{d-1} + ( \floor{A_d} + p_{max} + \lceil\tfrac{p_{max}}{n} \rceil), 1 \le d \le h \}$, $t_0 = 0$, and based on $SKS$, Algorithm SketchToSchedule can generate a feasible schedule with the makespan at most $(1+\epsilon)$ times the optimal makespan.
\end{theorem}
\begin{proof}
From the proof of Theorem~\ref{thm-alg3}, we know that 
the interval with the length $\floor{A_d} + p_{max}$ can feasibly fit in all the jobs of depth $d$ and with the process time at least $\tfrac{p_{max}}{n^2}$. If we add  additional length of $ \lceil n \cdot \tfrac{p_{max}}{n^2} \rceil = \lceil \tfrac{p_{max}}{n} \rceil$ to the interval, we can guarantee that both the large jobs and the small jobs of depth $d$ can be fit in. Hence, $SKS= \{t_d: t_d = t_{d-1} + ( \floor{A_d} + p_{max} + \lceil \tfrac{p_{max}}{n} \rceil), 1 \le d \le h \}$, $t_0 = 0$, describes a feasible schedule such that all the jobs of depth $d$ can be scheduled during the interval $[t_{d-1}, t_d]$.
Based on the sketch $SKS$, we can use Algorithm SketchToSchedule to generate  a feasible schedule with the makespan at most $$ t_h = \sum_{d=1}^h  (\floor{A_d} + p_{max} +\lceil \tfrac{p_{max}}{n} \rceil  )   
 \le \left(\sum_{d=1}^h  \floor{A_d}\right) +     (   h +1) p_{max}.$$
From the proof of Theorem~\ref{thm-alg3}, we know $t_h \le   C''_{max} +  (   h +1) p_{max} \le (1+\epsilon)C^*_{max}$.
\end{proof}

In summary, if we can read the input in two passes, based on the sketch of the schedule produced by all our streaming algorithms, the Algorithm SketchToSchedule can generate a feasible schedule with the makespan at most $(1+\epsilon)$ times the optimal value.

\begin{theorem}\label{thm: sketch-to-schedule}
For the  problems $P \mid prec, dp_j \le h,  p_{[min]} \le c \cdot p_{max} \mid C_{max}$, and $P \mid prec, dp_j \le h, p_{[n]} \le c \cdot p_{[ (1 - \alpha) n)]}  \mid C_{max}$,  when $ m   \le \tfrac  {2 n \epsilon} { 3 \cdot h \cdot c}$ and 
$ m  \le \tfrac { 2 n \alpha  \epsilon} { 3 (h+1) \cdot c} $, respectively,  there exist streaming approximation schemes that  can return an approximate value and a sketch of a schedule in one pass, and output a schedule for each job  in constant time   in  the second pass.
\end{theorem}

\section{Randomized Sublinear Time Algorithm}
In the previous section, we studied the streaming algorithms, which scan all the input data, generate a sketch of the input data, and use it to compute an approximate value of the optimal makespan and a sketch of a schedule that describes a feasible schedule with the approximated makespan. In this section, we study a different computing paradigm, sublinear time algorithms which are also inspired by the boost of multitude of data in manufacturing and service industry.
For sublinear time algorithms, our goal is to compute an approximate value of the optimal solution by   considering only a   fraction of the input  data. As most sublinear time algorithms,  our algorithms are randomized. 
Like streaming algorithms, our sublinear time algorithms also use  the sketch of the input to compute the approximate value and 
the sketch of the schedule.
The concept of the sketch of the input and the sketch of the schedule are similar to the ones that we defined in the streaming algorithms. However, since we do not read all input, the sketches are not accurate. We call them  {\bf estimated sketch of the input}, and {\bf estimated sketch of the schedule}. 

The {\bf estimated sketch of the input} is an estimated summary of the $n$ input jobs that is computed based on the sketch of   $n'$  sample jobs. The sample size $n'$ is determined by the approximation ratio $\epsilon$, and other parameters.  
We will show that with appropriate sample size, the estimated sketch of the input can give a good approximation of the accurate sketch of the input with high probability, and thus can give a good approximation of the optimal makespan. 
Formally, the {\bf estimated sketch of the input} is defined as follows:
\begin{definition}
For a given parameter $\delta$, and an instance of the problem $P \mid prec,  dp_j \le h, p_{max} \le c \cdot p_{min} \mid C_{max}$ or $P \mid prec, dp_j \le h, p_{[n]} \le c \cdot p_{[ (1 - \alpha) n)]} \mid C_{max}$, the \textbf{estimated sketch of the input} with respect to  $\delta$,   
is denoted as $\widehat{SKJ}_{\delta} = \{ (d, u, \hat{e}_{d,u})\}$ where $\hat{e}_{d,u}$ is the estimated number of jobs with the depth $d$ and the processing time in the range of $[(1+\delta)^u, (1+\delta)^{u+1})$. 
\end{definition}

Similarly, the {\bf estimated sketch of a schedule} is a concise description of a schedule. Based on the estimated sketch of the schedule, with high probability, we can generate a feasible schedule  with the makespan of at most $(1+\epsilon)$ times the optimal makespan. Formally the {\bf estimated sketch of a schedule} is defined as follows:

\begin{definition}
    For the problems $P \mid prec, dp_j \le h,  p_{max} \le c \cdot p_{max} \mid C_{max}$ and 
    $P \mid prec, dp_j \le h,  p_{[n]} \le c \cdot p_{[ (1 - \alpha) n)]} \mid C_{max}$, the {\bf estimated sketch of a schedule} describes a schedule 
    and consists of  a set of time instants $t_d$, $1 \le d \le h$, such that 
   all the jobs of depth $d$ are scheduled during the interval $[t_{d-1}, t_d)$ for $1 \le d \le h$ where $t_0 = 0$. Mathematically we denote the {\bf estimated sketch of a schedule} as $\widehat{SKS} = \{t_d: 1 \le d \le h\}$. 
\end{definition}

At the conceptual level, our  sublinear time algorithms have the following three steps:
\begin{itemize}
    \item[] \textbf{Step 1:} Compute the sample size $n'$ that is sublinear in $n$ but is sufficient for computing an estimated sketch of the input jobs that is  close to the accurate sketch for the original input data.  
    \item[] \textbf{Step 2:} Sample $n'$ jobs  uniformly at random from the input, find the sketch of the sampled jobs, $SKJ'_{\delta}=\{(d, u, n'_{d,u})\}$, and calculate the estimated sketch of all  $n$ jobs, $\widehat{SKJ}_{\delta} = \{(d, u, \hat{e}_{d,u})\}$.
    \item[] \textbf{Step 3:} Based on the estimated sketch of the input, compute an approximation of the optimal value and an estimate sketch of a schedule. 
\end{itemize}

In the following, we first develop a sublinear time algorithm for  $P \mid prec,  dp_j \le h, p_{max} \le c \cdot p_{ min} \mid C_{max}$ and then adapt it to solve the general problem $P \mid prec, dp_j \le h,  p_{[n]} \le c \cdot p_{[ (1 - \alpha) n)]} \mid C_{max}$. 

\subsection{Sublinear Time Algorithm  for  $P \mid  prec,  dp_j \le h, p_{max} \le c \cdot p_{ min} \mid C_{max}$ } 
 As before, each job  $j$ is represented by a pair $(p_j, dp_j)$. 
 Without loss of generality, we assume that $1 \le p_j \le c$ for all $1 \le j \le n$. Our algorithm mainly consists of the three steps as described above: (1) compute the sample size $n'$; (2) sample $n'$ jobs uniformly at random, find the sketch of the sampled jobs, $SKJ'_{\delta}=\{(d, u, n'_{d,u})\}$, where $n'_{d,u}$ is the number of sampled jobs with the depth $d$ and the processing time in the range of $[(1+\delta)^u, (1+\delta)^{u+1})$, and calculate the estimated sketch of the input for $n$ jobs, $\widehat{SKJ}_{\delta} = \{(d, u, \hat{e}_{d,u})\}$ such that $\hat{e}_{d,u}= \tfrac{n}{n'} n'_{d,u}$, and 
 $\hat{e}_{d,u}  \ge 2\tau(n, h, c)$, where $\tau(n, h, c)$ is determined by some parameters; (3) compute an approximation of the optimal value. The algorithm is formally described in Randomized-Algorithm1.

\begin{algorithm} 
\renewcommand{\thealgorithm}{}
\caption {Randomized-Algorithm1}
   Input:   Parameters: $m$,  $c$, $h$, $\epsilon$ 
   
  \hspace{0.4in}          Jobs:  $(p_j, dp_j)$, $1 \le j \le n$,  $1 \le dp_j \le h$
  
  Output:  An  approximation of the optimal makespan
\begin{algorithmic}[1]
\State compute the sample size $n'$
      \Indent
       \State let $\delta = \tfrac{\epsilon}{20}$, 
                 and  $k = \floor{ \log_{1+\delta} c   }$
      \State  let  $p = \tfrac{5 \delta}{2 c\cdot h\cdot k\cdot m} $, and $\beta =\delta p$
       \State  let  $n'  =  \tfrac{3} {  \beta ^2} \cdot {\ln {\tfrac{2}{\gamma}}}$, where $\gamma=\tfrac{1}{10hk}$
       \EndIndent
\State sample $ n'$  jobs uniformly at random, and compute  the sketch of the sampled jobs $SKJ'_{\delta}=\{(d, u, n'_{d,u})\}$ 
\State compute the estimated sketch of all jobs $\widehat{SKJ}_{\delta}$ 
\Indent 
 \State \label{line:random1-tau} let $\tau(n, h, c)= n \cdot p$ 
 \State  $\widehat{SKJ}_{\delta} = \emptyset$
 \For {each $ (d, u, n'_{d,u}) \in SKJ'_{\delta}$}
        \State \label{line:random1-e-d-u}let $\hat{e}_{d, u}=n \cdot \tfrac{n'_{d, u}}{n'}$
    \State  if $\hat{e}_{d, u} > 2\tau(n, h, c) $ 
      \State \hspace{0.2in} 
    $\widehat{SKJ}_{\delta} = \widehat{SKJ}_{\delta} \cup \{(d, u, \hat{e}_{d,u})\}$
 \EndFor
\EndIndent
\State \label{step: random1-compute-makespan} compute  the estimated    makespan
    \Indent
        \State let $rp_k = c$
        \State for each $u$, $0 \le u < k $ 
        \Indent \State let $rp_{u} =  (1 + \delta)^{u+1} $  
   \EndIndent
        \For {each $d$, $1 \le d \le h$}
         \State let $\hat{A}_d =      \frac{1}{m} \sum_{u=0}^k(  \hat{e}_{d,u}  \cdot rp_{u})$ ,    where $(d, u, \hat{e}_{d,u}) \in \widehat{SKJ}_{\delta}  $  
          \EndFor
        \State let $\hat{A} = \sum_{d=1}^h \left( \lfloor \hat{A}_d \rfloor  + c\right)$
    \EndIndent
\State  return $\hat{A}$  
\end{algorithmic}
\end{algorithm}
  
Now we give the performance analysis for the above  algorithm. The time complexity is dominated by the sampling operation. Thus we have the following lemma.  

\begin{lemma}\label{time-lemma}
The running time of the algorithm is $O(\tfrac{c^2 h^2 \log ^ 2 c } {\epsilon^6}  \log (\tfrac{h}{\epsilon} \log c) \cdot {  m^2}) $.
\end{lemma}

\begin{proof} The algorithm takes $n'$ random samples and the processing time for each sampled job is O(1). So the running time of the algorithm is
$O(n')=O(\tfrac{1}{\beta^2} \cdot \ln \tfrac{2}{\gamma})= O(\tfrac{c^2 h^2 k^2 m^2}{\epsilon^4} \log(hk)) =  O(\tfrac{c^2 h^2 \log ^ 2 c } {\epsilon^6}  \log (\tfrac{h}{\epsilon} \log c) \cdot {  m^2}) $.
\end{proof}

From now on we focus on the accuracy analysis for our algorithm. Since in our analysis we use the bounds that  Ma  \cite{ma00} has obtained based on  the well-known Chernoff bounds (see~\cite{mrp13}), and the  union bound from probability theory, we list them in the following for reference.

\begin{lemma}[Lemma 3 in Ma \cite{ma00}]\label{chernoff-lemma}
Let $X_1,\ldots , X_n$ be $n$ independent random $0$-$1$ variables and $X=\sum_{i=1}^n X_i$.
\begin{enumerate}

 \item  If $X_i$ takes $1$ with probability at most $p$ for $i=1,\ldots ,
n$, then for any $\beta>0$, $\Pr(X>pn+\beta
n)<e^{-\tfrac{1}{ 3}n\beta^2}$.

\item If  $X_i$ takes $1$ with probability at least $p$ for
$i=1,\ldots , n$, then for any $\beta>0$, $\Pr(X<pn-\beta
n)<e^{- \tfrac{1}{  2}n\beta^2}$.
\end{enumerate}
\end{lemma}

\begin{fact}[Union bound]\label{union-bound}
Let $E_1,E_2,\ldots, E_m$ be $m$ events that may not be independent, we have the inequality $$\Pr(E_1\cup E_2 \ldots \cup E_m)\le \Pr(E_1)+\Pr(E_2)+\ldots+\Pr(E_m).$$
\end{fact}

We will use Lemma~\ref{chernoff-lemma} and Fact~\ref{union-bound} to show that $\hat{e}_{d,u}$, computed by Randomized-Algorithm1, is a good estimate of the exact number  of jobs with the depth $d$ and  processing time in the range of $[(1+\delta)^u, (1+\delta)^{u+1})$, $n_{d,u}$.
Specifically, we have that with high probability: (1) if   $n_{d,u}$ is at least $\tau(n,h,c)$, then our estimate, $\hat{e}_{d,u}$, is in the range of $[(1-\delta)n_{d,u}, (1+\delta)n_{d,u}]$; and (2) if $n_{d,u} < \tau(n,h,c)$, our estimated number of jobs, $\hat{e}_{d,u}$, is no more than $2\tau(n,h,c)$.

\begin{lemma} \label{lemma:prob-e-d-u-is-good}
For any $d$, $u$, let  $\hat{e}_{d,u}$ be the value computed by Randomized-Algorithm1, then we have: 
\begin{enumerate}
\item[(i)] If $n_{d, u}\ge \tau(n,h,c)$, 
$\Pr( (1-\delta)n_{d, u}\le \hat{e}_{d, u}\le (1+\delta)n_{d, u}) \ge 1-\gamma$; and 
\item [(ii)]
If $n_{d, u}< \tau(n,h,c)$, 
$\Pr(\hat{e}_{d, u} \le 2\tau(n,h, c)) \ge 1 - \gamma$.
\end{enumerate}
\end{lemma}

\begin{proof}
Let $X_{i}$ denote the  indicator random variable for the event that   the   $i$-th sample job has depth $d$, and the processing time is in $[(1+\delta)^u, (1+\delta)^{u+1})$. Then  $n'_{d,u} = \sum_{i=1}^{n'} X_i$.  
Since $n'$ jobs are   sampled uniformly at random from $n$ jobs, we have $\Pr(X_i = 1) = \tfrac{n_{d, u}}{n}$. 
For convenience, we let $p_0=\tfrac{n_{d, u}}{n}$. By line~\emph{\ref{line:random1-tau}} of our algorithm, $p = \tfrac{\tau(n,h,c)}{n}$.
 
We first prove (i): if $n_{d, u}\ge \tau(n,h,c)$, 
$\Pr( (1-\delta)n_{d, u}\le \hat{e}_{d, u}\le (1+\delta)n_{d, u}) \ge 1-\gamma$. It is sufficient to show that $\Pr( \hat{e}_{d, u} 
 \le (1-\delta)n_{d, u}) \le \tfrac{\gamma}{2}$ and   $Pr( \hat{e}_{d, u} \ge(1+\delta)n_{d, u}) \le \tfrac{\gamma}{2}$. 
 By  line~\emph{\ref{line:random1-e-d-u}} of the algorithm, $\hat{e}_{d,u} = n \cdot \tfrac{n'_{d,u}}{n'}$, thus we have
\begin{eqnarray*}
\Pr( \hat{e}_{d,u}   \le   (1 - \delta) n_{d,u})  
  &  =  &   \Pr( n \cdot \tfrac{n'_{d,u}}{n'}   \le   (1 - \delta) n_{d,u})    \\
 & = & \Pr( n'_{d,u} \le (1 - \delta)  \tfrac{n_{d, u}}{n} \cdot n') \\
 & = & \Pr( n'_{d,u} \le (1 - \delta)  p_0 {n'} ) \\
 & = & \Pr( n'_{d,u} \le (p_0 - \delta p_0)   {n'}) 
\end{eqnarray*}
If $n_{d, u}\ge \tau(n,h,c)$, then $Pr(X_i = 1) = p_0 = \tfrac{n_{d,u}}{n} \ge \tfrac{\tau(n,h,c)}{n} = p$.  Using this fact, and applying Lemma~\ref{chernoff-lemma} for the variable  $n'_{d,u}$, $n'_{d,u} = \sum_{i=1}^{n'} X_i$, we get 
$$ \Pr( n'_{d,u} \le (p_0 - \delta p_0) n')    \le   e^{-\tfrac{1}{2}n'(\delta p_0)^2}  
\le  e^{-\tfrac{1}{2}n'(\delta p)^2}   
 \le   e^{-\tfrac{1}{2}n'\beta^2} \le \tfrac{\gamma}{2},$$  which means that $   \Pr( \hat{e}_{d,u}   \le   (1 - \delta) n_{d,u}) \le \tfrac{\gamma}{2}$.   Similarly, we have 
\begin{eqnarray*}
\Pr( \hat{e}_{d,u}  \ge   (1 + \delta) n_{d,u}) 
& = & \Pr( n'_{d,u} \ge (p_0+ \delta p_0) n'  ) \\
& \le &  e^{-\tfrac{1}{3} n' (\delta p_0)^2} \\ 
& \le & e^{-\tfrac{1}{3}n'(\delta p)^2}  \\  
& \le & e^{-\tfrac{1}{3}n'\beta^2}\\
& \le & \tfrac{\gamma}{2}.
\end{eqnarray*}

  Next, we show  (ii): if $n_{d, u}< \tau(n,h,c)$,  $\Pr(\hat{e}_{d, u} \le 2\tau(n,h, c)) \ge 1 - \gamma$. As for (i), we prove that $\Pr(\hat{e}_{d, u} > 2\tau(n,h, c)) \le   \gamma$. By line~\emph{\ref{line:random1-tau}} of the algorithm, $ \tau(n,h,c) =   n \cdot p$, and $\hat{e}_{d,u} = n \cdot \tfrac{n'_{d,u}}{n'}$. 
$$\Pr( \hat{e}_{d,u}  >  2\tau(n,h, c))    =    \Pr( \hat{e}_{d,u}  >  2  n p ) 
  =   \Pr( n'_{d,u} > 2   n' p)  
   \le   \Pr( n'_{d,u} > (p + \delta p )   n' ) .$$
If $n_{d, u} < \tau(n,h,c)$, then $Pr(X_i = 1) =  \tfrac{n_{d,u}}{n} \le \tfrac{\tau(n,h,c)}{n} = p$.  Using this fact, and applying Lemma~\ref{chernoff-lemma} for the variable  $n'_{d,u}$, $n'_{d,u} = \sum_{i=1}^{n'} X_i$, we get 
$$ \Pr( n'_{d,u} > (p + \delta p )   n' )   
  \le    e^{-\beta^2 \tfrac{n'} {3}} \le \tfrac{\gamma}{2},$$ which implies that   
$\Pr( \hat{e}_{d,u}  >  2\tau(n,h, c)) \le \tfrac{\gamma}{2} < \gamma$.  This completes the proof.
\end{proof}
Lemma~\ref{lemma:prob-e-d-u-is-good} tells us that the estimated sketch of input $\widehat{SKJ}_{\delta}$ approximates the exact sketch of input $SKJ_{\delta}$ very well. Based on this, we will show that the estimated   makespan, $\hat{A}$, computed from the estimated sketch, is a good approximation of the optimal makespan. For the ease of  our analysis and proof later, we summarize all the symbols we use in the following: 
\begin{itemize}
    \item $I$: the input  instance 
    for the algorithm
    \item $SKJ_\delta = \{(d,u, n_{d,u})\}$: the exact sketch of all jobs in $I$ where $n_{d,u}$ is the number of jobs in $I$ with the depth $d$ and the processing time in the range of $[(1+\delta)^u, (1+\delta)^{u+1})$ for all $1 \le d \le h$ and $1 \le u \le k$
    \item $I_{round}$: the instance corresponding to the sketch $SKJ_\delta$ with the rounded processing times for all the jobs, that is, for each $(d, u, n_{d,u}) \in SKJ_{\delta}$, there are $n_{d,u}$ jobs at depth $d$ whose processing times are  $rp_u$ 
         \item $I_{big}$: the instance obtained from the instance $I_{round}$ by removing the jobs corresponding to $(d,u, n_{d,u})$ where  $n_{d,u}<3\tau(n, h, c)$ for all $1 \le d \le h$ and $1 \le u \le k$
    \item $\widehat{SKJ}_\delta = \{(d,u,\hat{e}_{d,u})\}$: the estimated sketch for the jobs in $I$, which is computed by Randomized-Algorithm1, and where $\hat{e}_{d,u}$ is the estimated value for $n_{d,u}$. Note that only the tuples with $\hat{e}_{d,u}  > 2 \tau (n,h,c)$ are included in $\widehat{SKJ}_\delta$.
    \item $\hat{I}$: the  instance  corresponding to the estimated sketch $\widehat{SKJ}_\delta = \{(d,u,\hat{e}_{d,u})\}$, that is, for each $(d, u, \hat{e}_{d,u}) \in \widehat{SKJ}_{\delta}$, there are $\hat{e}_{d,u}$  jobs at depth $d$ whose processing times are  $rp_u $
\item $(d,u)$-group of an instance: the group of all the jobs in the instance with depth $d$ and processing time $rp_u$ 
\end{itemize}

We first compare the   optimal makespan of instance $\hat{I}$  and that of  instance $I_{big}$.  
\begin{lemma}\label{lemma:rounded-big-vs-estimated}
Let $C_{max}^*(I_{big})$ and $C_{max}^*(\hat{I})$ be the optimal makespan for instances $I_{big}$ and $\hat{I}$ respectively, with   probability of at least $\tfrac{9}{10}$, we have

\begin{equation}\label{eq:rounded_large_original_vs_sample_approx}
(1 - \delta) ( C_{max}^*(I_{big}) - h \cdot c) <    C_{max}^*(\hat{I}) \le (1 + \delta) C_{max}^*(I_{big}) + \tfrac{15 \delta n } {m} + h \cdot c.
\end{equation}
\end{lemma}

\begin{proof}    
From our definition of $I_{big}$ and $\hat{I}$, we know that for any $(d,u)$-group included in $I_{big}$, we must have $n_{d,u} \ge 3 \tau(n, h, c)$ and for any $(d,u)$-group included in $\hat{I}$, we must have $\hat{e}_{d,u} > 2 \tau(n, h, c)$.
We first show that with high probability the instance $\hat{I}$ contains all jobs from instance $I_{big}$. 
Consider an arbitrary $(d,u)$-group from $I_{big}$, we must have 
 $n_{d,u} \ge 3 \tau(n, h, c)$, since $\delta = \tfrac{\epsilon}{20} < \tfrac{1}{20}$, we have $(1-\delta)n_{d, u} \ge 2  \tau(n, h, c)$. 
By Lemma~\ref{lemma:prob-e-d-u-is-good},  with the probability of at least $1-\gamma$, we have $ 2 \tau(n, h, c) < (1-\delta)n_{d, u}\le \hat{e}_{d, u}\le (1+\delta)n_{d, u}$.
That means, with the probability of at least $1-\gamma$, we have that any $(d, u)$-group in $I_{big}$ is also included in $\hat{I}$. In other words, the probability that a $(d, u)$-group is in $I_{big}$ but not in $\hat{I}$ is less than $\gamma$. 
Since there are at most $h \cdot k$ $(d, u)$-groups, by Fact~\ref{union-bound}, the probability that some $(d, u)$-groups are in $I_{big}$ but not in $\hat{I}$ is at most  $\gamma\cdot h\cdot k = \tfrac{1}{10}$. Therefore, considering all $(d, u)$-groups in $I_{big}$, we have that with probability at least $\tfrac{9}{10}$, all $(d, u)$-groups that are included in $I_{big}$ are also included in $\hat{I}$. 

A lower bound of  $C_{max}^*(\hat{I})$ can be obtained by considering only those $(d,u)$-groups that are in $I_{big}$. To schedule  the jobs in these groups from $\hat{I}$, one need an interval of length at least 
$\sum_d \tfrac{1}{m} \sum_u (\hat{e}_{d, u} \cdot rp_{u}) \ge \sum_d \tfrac{1}{m} \sum_u ((1-\delta) n_{d, u} \cdot rp_{u})$. So we have 
\begin{equation}\label{I-hat-lower-bound-1}C_{max}^*(\hat{I}) 
\ge \sum_{d=1}^h \left(  \frac{1}{m}\sum_{u=0}^k \left((1-\delta) n_{d, u} \cdot rp_{u} \right) \right). 
\end{equation}
For all the jobs from $I_{big}$, we have:
\begin{equation}\label{I-big-lower-bound-1} C_{max}^*(I_{big})  \ge \sum_{d=1}^h \left( \tfrac{1}{m} \sum_{u=0}^k (  n_{d, u} \cdot rp_{u})  \right), \end{equation}
  and 
\begin{equation}  \label{I-big-upper-bound-1}  C_{max}^*(I_{big}) \le \sum_{d=1}^h \left( \tfrac{1}{m} \sum_{u=0}^k (  n_{d, u} \cdot rp_{u}) + c \right) = \sum_{d=1}^h \left( \tfrac{1}{m} \sum_{u=0}^k (  n_{d, u} \cdot rp_{u}) \right) + h \cdot c.\end{equation}
By inequalities~(\ref{I-hat-lower-bound-1}) and (\ref{I-big-upper-bound-1}) we have
 \begin{equation}  \label{I-hat-lower-bound-2}
 C_{max}^*(\hat{I}) \ge (1 - \delta) ( C_{max}^*(I_{big}) - h \cdot c). 
 \end{equation}

Next, we consider the upper bound of $C_{max}^*(\hat{I})$. We split the jobs in $\hat{I}$ into two parts:  
those $(d, u)$-groups that are  in both $I_{big}$ and $\hat{I}$ , and those  $(d, u)$-groups that are in $\hat{I}$ but not in $I_{big}$.
For the jobs in $\hat{I}$ from the former,  we need an interval of length at most 
$\sum_d (\tfrac{1}{m} ( \sum_u (\hat{e}_{d, u} \cdot rp_{u})) + c) \le \sum_d (\tfrac{1}{m} (\sum_u ((1+\delta) n_{d, u} \cdot rp_{u})) + c)$ to schedule them; for the jobs from the latter $(d,u)$-groups, we  
note that   
  each such group must correspond to a group  in instance $I$ where $n_{d, u} <  3 \tau(n, h, c)$   and there are at most $h \cdot k$ such groups. By Lemma~\ref{lemma:prob-e-d-u-is-good},  with the probability of at least $1-\gamma$, we have at most $ \hat{e}_{d, u}\le 6 \tau(n, h, c)$ jobs in $\hat{I}$ for each $(d,u)$-group in the latter type. 
Thus we can schedule these jobs in an interval of at most $  6 \tau(n,h,c) \cdot h   \cdot k \cdot c $.
Combining both types of groups and by inequality~(\ref{I-big-lower-bound-1}), we have 
\begin{eqnarray*}
C_{max}^*(\hat{I})  
&\le & \sum_{d=1}^h \left( \tfrac{1}{m}\left(\sum_{u=0}^k \left(  (1 + \delta) n_{d, u} \cdot rp_{u} \right) \right) + c \right)  +   6 \tau(n,h,c)  \cdot h \cdot k \cdot  c \\
&\le &  (1 + \delta) C_{max}^*(I_{big})+h \cdot c +   6 \tau(n,h,c)  \cdot h  \cdot k   \cdot  c. 
\end{eqnarray*}
By line~\ref{line:random1-tau} of the algorithm, $\tau(n,h,c) =  {n}\cdot p = \tfrac {  5 \delta n }{ 2 c \cdot h \cdot k \cdot m}$,
we have
\begin{equation}\label{I-hat-upper-bound-1}
C_{max}^*(\hat{I}) \le (1 + \delta) C_{max}^*(I_{big})  + \tfrac{15 \delta n } {m} + h \cdot c.
\end{equation}
Therefore, from (\ref{I-hat-lower-bound-2}) and (\ref{I-hat-upper-bound-1}), we get
$$ (1 - \delta) ( C_{max}^*(I_{big}) - h \cdot c) \le    C_{max}^*(\hat{I}) \le (1 + \delta) C_{max}^*(I_{big})  + \tfrac{15 \delta n } {m} + h \cdot c.$$
\end{proof}
The next lemma   compares the   optimal makespan of instance  $I_{round}$  and that of  instance $I$ and $I_{big}$.
\begin{lemma}\label{lemma:rounded-vs-optimal}
Let $C_{max}^*(I)$ and $C_{max}^*(I_{round})$ be the optimal makespan for instances $I$ and $I_{round}$ respectively, we have
the following inequalities:
\begin{equation}\label{eq:rounding_1}
  C_{max}^*(I) \le  C_{max}^*(I_{round})  \le  (1 + \delta) C_{max}^*(I).
\end{equation}
\begin{equation}\label{eq:rounded_large_groups}
   C_{max}^*(I_{round}) - \tfrac{8 \delta n } {m}   \le  C_{max}^*(I_{big}) \le  C_{max}^*(I_{round}).
\end{equation} 
\end{lemma}

\begin{proof} By our notation, $I$ is the original instance of $n$ jobs where a job $j$ has processing time $p_j$ and depth $dp_j$, and $I_{round}$ is the  instance from $I$  after rounding up the jobs' processing time such that if $(1+ \delta)^ u \le p_j \le (1 + \delta)^{u+1}$, then the rounded processing time is $rp_u \le (1 + \delta)p_j$ . It is easy to see that  we have
$$C_{max}^*(I) \le  C_{max}^*(I_{round}) < (1 + \delta) C_{max}^*(I).$$
The instance $I_{big}$ can be obtained from the instance $I_{round}$ by removing those $(d,u)$-group jobs where $n_{d,u} < 3 \tau(n,h,c)$. The total number  of the jobs removed is at most $3 \tau(n,h,c)  \cdot h \cdot k$, and each of these jobs have processing time at most $c$.  Therefore,  we have $  C_{max}^*(I_{round}) - { 3 \tau(n,h,c)   \cdot h  \cdot k \cdot c } \le  C_{max}^*(I_{big}) \le  C_{max}^*(I_{round}).$ Since $\tau(n,h,c) =  \tfrac {  5 \delta n }{ 2 c \cdot h \cdot k \cdot m}$, we get 
 $$  C_{max}^*(I_{round}) - \tfrac{8 \delta n } {m}  \le  C_{max}^*(I_{big}) \le  C_{max}^*(I_{round}).$$
\end{proof}
Combining all the lemmas we proved in this section, we can prove that the Randomized-Algorithm1 is an approximation scheme. 
\begin{theorem}\label{thm:const-alg}
If $m \le  \tfrac {n \epsilon}{20h \cdot c}$, Randomized-Algorithm1  is a randomized  $(1+\epsilon)$-approximation scheme for $P \mid prec,  dp_j \le h, p_{max} \le c \cdot p_{min} \mid C_{max}$	
 that runs in  $O(\tfrac{c^2 h^2 \log ^ 2 c } {\epsilon^6}  \log (\tfrac{h}{\epsilon} \log c) \cdot {  m^2}) $
 time.
\end{theorem}

\begin{proof}
The running time follows from Lemma~\ref{time-lemma}. We focus on the approximation ratio.
By our notation, $\hat{I}$ is the instance corresponding to the estimated sketch $\widehat{SKJ}_\delta = \{(d,u,\hat{e}_{d,u})\}$ where $\hat{e}_{d,u}$ is the estimated value for $n_{d,u}$. Only the tuples with $\hat{e}_{d,u}  > 2 \tau (n,h,c)$ are included in $\widehat{SKJ}_\delta$. By Randomized-Algorithm1, 
$\hat{A}_d = \frac{1}{m} \sum_{u=0}^k(  \hat{e}_{d,u}  \cdot rp_{u})$ and $\hat{A} = \sum_{d=1}^h ( \lfloor \hat{A}_d \rfloor  + c )$. Following the same proof as in Theorem~\ref{thm-alg1}, we can get 
 $$C_{max}^*(\hat{I}) \le \hat{A} \le  C_{max}^*(\hat{I}) +   h \cdot c.$$
  By inequality (\ref{eq:rounded_large_original_vs_sample_approx}), we get, with probability at least $\tfrac{9}{10}$  
 $$\hat{A}   \le      C_{max}^*(\hat{I}) + h \cdot c   \le  (1 + \delta) C_{max}^*(I_{big})  + \tfrac{15 \delta n } {m} + 2h \cdot c.$$
  If $m \le  \tfrac {n \epsilon}{20h \cdot c}$, with 
  $\delta=\tfrac{\epsilon}{20}$, we get $h \cdot c \le \tfrac{\delta n}{m}$.
Thus, 
we get, 
\begin{eqnarray*} 
\hat{A}
 & \le &   (1 + \delta) C_{max}^*(I_{big})  + \tfrac{15 \delta n } {m} + 2h \cdot c \\
& \le & (1 + \delta)   C_{max}^*(I_{big}) + \tfrac{17 \delta n } { m} \\ 
& \le &  (1 + \delta) C_{max}^*(I_{round})  + \tfrac{17 \delta n }{m}  \hspace{1.0in} \text { by} (\ref{eq:rounded_large_groups})  \\
& \le &  (1 + \delta)^2 C_{max}^*(I)  + \tfrac{17 \delta n }{m} \hspace{1.2 in } \text { by } (\ref{eq:rounding_1}) \\
& \le &  (1 + 20 \delta)  C_{max}^*(I)  \hspace{1.6 in } \text { by }  C_{max}^*(I) \ge \tfrac{n}{m}   \\ 
& \le & (1 + \epsilon) C_{max}^*(I),   \hspace{1.7in }\text { by } \delta =  \tfrac{\epsilon }{20 }
\end{eqnarray*}
and
\begin{eqnarray*} 
\hat{A} &  \ge  & C_{max}^*(\hat{I}) \\
&  \ge  & (1 - \delta) ( C_{max}^*(I_{big}) - h \cdot c) \hspace{0.9 in } \text { by }(\ref {eq:rounded_large_original_vs_sample_approx})\\
&  \ge  & (1 - \delta) ( C_{max}^*(I_{big}) -   \tfrac{\delta n}{m} ) \hspace{1.0 in } \text { by }  h \cdot c \le \tfrac{\delta n}{m}  \\
&  \ge  & ( 1 - \delta) (C_{max}^*(I_{round}) - \tfrac{ 9 \delta n}{m}  )  \hspace{0.8in} \text { by} (\ref{eq:rounded_large_groups}) \\
&  \ge & ( 1 - \delta) (   { C_{max}^*(I) }  - \tfrac{ 9 \delta n}{m} )   \hspace{1.1in} \text { by} (\ref{eq:rounding_1})\\
& \ge &   (   1 - \delta  ) (1 -  9 \delta)  C_{max}^*(I) \hspace{1.1in } \text { by } C_{max}^*(I) \ge \tfrac{n}{m}\\
& \ge &   (  1 -  10 \delta)C_{max}^*(I) \hspace{1.5in }  \text { by } \delta =  \tfrac{\epsilon }{20 }\\ 
& \ge &   (  1 -  \epsilon )C_{max}^*(I)   
\end{eqnarray*}

\end{proof}

Based on the above theorem,  when $m = {o}(n^{ 1/2})$, Randomized-Algorithm is  a sublinear time approximation scheme.
\begin{corollary}\label{cor:sublinear-time1} When
$m = {o}(n^{ 1/2})$, Randomized-Algorithm1  is a randomized  $(1+\epsilon)$-approximation scheme for $P \mid prec,  dp_j \le h, p_{max} \le c \cdot p_{min} \mid C_{max}$	
 that runs in  sublinear time.
\end{corollary}

\subsection{Sublinear Time Algorithm for $P \mid prec, dp_j \le h,  p_{[n]} \le c \cdot p_{[ (1 - \alpha) n)]} \mid C_{max}$ }

In this section, we will generalize  Randomized-Algorithm1  to solve  the general problem $P \mid prec, dp_j \le h,  p_{[n]} \le c \cdot p_{[ (1 - \alpha) n)]} \mid C_{max}$. The idea is basically similar except some pre-processing is needed because we do not know the processing time range of the top $\alpha n$ jobs. Specifically, our sublinear algorithm first samples some jobs to determine the upper bound of  the largest   job, 
then samples enough number of jobs to generate the estimated sketch of the input, and finally computes the approximation of the optimal value based on the estimated sketch of the input. The details are given in Randomized-Algorithm2. 
\begin{algorithm} 
\renewcommand{\thealgorithm}{}
\caption {Randomized-Algorithm2}
   Input:   Parameters  $m$,  $c$, $h$, $\epsilon$, $\alpha$ 
   
   \hspace{0.4in}   Jobs:  $(p_j, dp_j)$, $1 \le j \le n$,  $1 \le dp_j \le h$     
   
  Output:  An  approximation of the optimal makespan
\begin{algorithmic}[1]
\State determine the upper bound of  the largest   job: 
  \Indent
   \State  let $\delta = \tfrac{\epsilon}{20}$,  $k = \floor{ \log_{1+\delta} \tfrac{c n } {\delta}  }$, and  $\gamma= \tfrac{1}{ 10hk}$
  \State let $n_0=1$ if $\alpha=1$, and $n_0=\ceiling{ \tfrac{ \ln \gamma}{    \ln (1-\alpha)}}$ if $\alpha<1$
   \State sample $n_0$ jobs uniformly at random 
   \State let $w_0$ be the largest processing time among all  the $n_0$ sampled jobs 
   \EndIndent
\State determine the sample size $n'$: 
  \Indent
  \State let  $p = \tfrac{ 5\alpha\delta }{2 c^2\cdot h \cdot k\cdot m}$, and 
    $\beta =\delta p$ 
   \State   let  $n'  =   \tfrac{3} { \alpha \beta ^2} \cdot {\ln \tfrac{2}{ \gamma}}$   
   \EndIndent 
\State sample $n'$ jobs uniformly at random
\State remove those jobs   whose processing time is at most $\tfrac{\delta w_0}{n}$ from the sampled jobs
 \State compute the sketch of the remaining sample jobs $SKJ'_{\delta} = \{(d, u, n'_{d,u})\}$
 
\State compute the estimated sketch of all jobs $\widehat{SKJ}_{\delta}$ 
\Indent
 \State \label{line:random2-tau}  let $\tau(n, h, c)={n}\cdot p $
 \State  $\widehat{SKJ}_{\delta} = \emptyset$
 \For { each $ (d, u, n'_{d,u}) \in SKJ'_{\delta}$ }
   \State  let $\hat{e}_{d, u}=n \cdot \tfrac{n'_{d, u}}{n'}$
   \State   if $\hat{e}_{d, u} > 2\tau(n, h, c) $ 
   \State  \hspace{0.2in} 
    $\widehat{SKJ}_{\delta} = \widehat{SKJ}_{\delta} \cup \{(d, u, \hat{e}_{d,u})\}$
\EndFor
\EndIndent
\State compute  the estimated    makespan
      \Indent   
        \State  let $u_{-} = \floor{\log_{1+\delta} \tfrac{\delta w_0}{n}}$,  $u_{+} =\floor{\log_{1+\delta} c w_0}$     
      \State  let $rp_{u_{+}} = c w_0$
        \State   for each $ u_{-} \le  u < u_{+} $ 
        \Indent 
        \State let  $rp_{u} =  (1 + \delta)^{u+1} $
        \EndIndent
       \State  for each $d$, $1 \le d \le h$
       \State  \hspace{0.2in} let $\hat{A}_d =      \frac{1}{m} \sum_{u=u_{-}}^{u_{+}}(  \hat{e}_{d,u}  \cdot rp_{u})$ ,    where $(d, u, \hat{e}_{d,u}) \in \widehat{SKJ}_{\delta}  $   
       \State  $\hat{A} = \sum_{d=1}^h  \left(\floor{\hat{A}_d} + c w_0 \right )$
   \EndIndent
  \State  return $\hat{A}$  
\end{algorithmic}
\end{algorithm}

 Like the Randomized-Algorithm1,   the time complexity of Randomized-Algorithm2 is dominated by the sample size $n'$. However, $n'$ in this algorithm depends on $n$. Still, we will show in the lemma below that the   running time of the algorithm is sublinear when $m = {o}(n^{ 1/2})$. 
 \begin{lemma}\label{lemma:sublinear-time}
Randomized-Algorithm2 runs in time  $O( \tfrac{c^4 h^2} {\alpha^3 \epsilon^6} \cdot  m^2 \log ^ 2 (\tfrac{c n}{\epsilon})  \log (\tfrac{h}{\epsilon} \log (\tfrac{c n}{\epsilon}  )))$.
\end{lemma}

\begin{proof}
The running time is dominated by the sampling of  $n_0+n' = O(n')$ jobs. Thus its running time is 
 $$  O(n')=O(\tfrac{1}{\alpha \beta^2} \cdot \ln \tfrac{2}{\gamma}))= 
 O( \tfrac{1}{ \alpha \delta^2} (\tfrac{c^2 h  k  m }{ \alpha \delta} )^2 \log hk) =   O( \tfrac{c^4 h^2} {\alpha^3 \epsilon^6} \cdot m^2  \log ^ 2 (\tfrac{c n}{\epsilon})  \log (\tfrac{h}{\epsilon} \log (\tfrac{c n}{\epsilon}  ))).$$  
\end{proof}
 
The next lemma shows that by sampling $n_0$ jobs, we can get a good estimate of the largest processing time  $p_{[n]}$. 
\begin{lemma}\label{I0-lemma}\scrod
 With probability at least $1-\gamma$, $p_{[n]} \le c w_0$.  
\end{lemma}

\begin{proof}
Since we sample the jobs uniformly, the probability that a  job   from the top $\alpha n$ largest jobs is selected is $\alpha$. The probability that no job from top $\alpha n$ largest jobs is sampled is at most   $(1- \alpha)^{n_0}\le \gamma$, which implies  that with probability of at least $1-\gamma$, some jobs from the top $\alpha n$ largest jobs are sampled, which means  $w_0 \ge p_{[1-\alpha)n]}$ and $p_[n] \le c w_0$.
\end{proof}

The next   lemma is similar to Lemma~\ref{lemma:prob-e-d-u-is-good} which states that $\hat{e}_{d, u}$ is a good estimate of  the number of corresponding jobs in the input instance,  $n_{d, u}$. 
The only difference is that here we focus on the jobs whose processing time is at least $\tfrac { \delta w_0}{n}$. 
The proof is the same and we omit here. 

 \begin{lemma}
 For any $d$, $u$, let  $\hat{e}_{d,u}$ be the value computed by Randomized-Algorithm2, then we have: 
\begin{enumerate}
\item[(i)] If $n_{d, u}\ge \tau(n,h,c)$, 
$\Pr( (1-\delta)n_{d, u}\le \hat{e}_{d, u}\le (1+\delta)n_{d, u}) \ge 1-\gamma$, and 
\item[(ii)]  If $n_{d, u}< \tau(n,h,c)$,  $\Pr(\hat{e}_{d, u} \le 2\tau(n,h, c)) \ge 1 - \gamma$.
\end{enumerate}
\end{lemma}

Like Theorem~\ref{thm:const-alg}, we can prove that Randomized-Algorithm2 is an approximation scheme.
\begin{theorem}\label{thm-sublinear-time-alg} For $P \mid prec, dp_j \le h,  p_{[n]} \le c \cdot p_{[ (1 - \alpha) n)]} \mid C_{max}$, when $ m  \le    \tfrac { n \cdot \alpha \cdot \epsilon}{20 c^2  \cdot h  } $, 
Randomized-Algorithm2 is a randomized $(1+\epsilon)$-approximation scheme  
that runs in time  $O( \tfrac{c^4 h^2 } {\alpha^3 \epsilon^6} \cdot m^2 \log ^ 2 (\tfrac{c n}{\epsilon})  \log (\tfrac{h}{\epsilon} \log (\tfrac{c n}{\epsilon}  )))$.
\end{theorem}

\begin{proof}
The running time follows from Lemma~\ref{lemma:sublinear-time}.
We  consider   the approximation ratio only. The proof is   similar to that of Theorem~\ref{thm:const-alg}.  
For the input instance $I$, let $I_{round}$ be the instance obtained from $I$ by rounding up the processing times for the jobs with $p_j \ge \tfrac{\delta w_0}{n}$. Let $C_{max}^*(I) $ and $ C_{max}^*(I_{round})$ respectively be the optimal makespan for $I$ and $I_{round}$.
Then we still have the same inequalities between $C_{max}^*(I_{round}) $ and $ C_{max}^*(I)$: 
\begin{equation}\label{round-vs original}
     C_{max}^*(I)  \le  C_{max}^*(I_{round}) \le (1+\delta)C_{max}^*(I)
\end{equation}
Let $I_{big}$ be the instance obtained from the instance $I_{round}$ by removing not only  the $(d,u)$-groups with $n_{d,u}< 3\tau(n, h, c)$  but also the groups of the jobs whose processing time is less than $ \tfrac {\delta w_0} {n}$.  The total processing time of the   
jobs with processing time less than 
$ \tfrac {\delta w_0} {n}$ is at most  $n \cdot \tfrac {\delta w_0} {n} \le \delta w_0 \le \delta  C_{max}^*(I_{round})$.   
The other jobs removed   belong to the groups with $n_{d,u}< 3\tau(n, h, c)$, and each of these jobs has processing time  at least $ \tfrac {\delta w_0} {n}$ and   at  most $c w_0$.   There are at most  $ h \cdot k$ such groups where $  k =  \log_{1+\delta} \tfrac{c n}{\delta}$ as defined in the algorithm. 
Thus the total processing time of these jobs is at most
$$  3\tau(n, h, c)   \cdot h  \cdot  k \cdot c w_0  =   3 n \cdot \tfrac{5 \alpha \delta}{2 c^2 \cdot  h \cdot  k \cdot  m}    \cdot h \cdot k \cdot c w_0\le  \tfrac{8  \delta \cdot  \alpha \cdot n w_0}{ c m}. $$  
Thus we have 
 \begin{equation}\label{round-vs-big-alpha}
 C_{max}^*(I_{round})  - \delta  C_{max}^*(I_{round} ) -  \tfrac{8  \delta \cdot  \alpha \cdot n w_0}{c m}\le C_{max}^*(I_{big}) \le C_{max}^*(I_{round})
 \end{equation}

 As before, let $\hat{I}$ be  the instance corresponding to the sketch $\widehat{SKJ}_\delta$, which contains $\hat{e}_{d,u}$ number of jobs with the depth of $d$ and the processing time $ rp_u $ where   $\hat{e}_{d, u} > 2\tau(n, h, c)$. Then the optimal makespan of $\hat{I}$, $C_{max}^*(\hat{I})$,  is at least $\hat{A}_d =      \frac{1}{m} \sum_{u=u_{-}}^{u_{+}}(  \hat{e}_{d,u}  \cdot rp_{u})$.
 Between the $I_{big}$ and $\hat{I}$, we can use similar argument for (\ref{I-hat-lower-bound-2}) and (\ref{I-hat-upper-bound-1}) to show that with probability at least $\tfrac{9}{10}$, we have 
  \begin{equation}\label{big-vs-estimate-1-alpha}
(1 - \delta) ( C_{max}^*(I_{big}) - h \cdot cw_0) <    C_{max}^*(\hat{I}) 
\end{equation}
and 
  \begin{equation}\label{big-vs-estimate-2-alpha}
  C_{max}^*(\hat{I}) \le (1 + \delta) C_{max}^*(I_{big})  + \tfrac{15 \delta \alpha n w_0} {cm} + h \cdot cw_0
\end{equation}
The returned value  $\hat{A} =\sum_{d=1}^h  \left(\lfloor\hat{A}_d \rfloor + c w_0 \right )$ is at least $C_{max}^*(\hat{I})$  
and   
$$ \hat{A}  =   \sum_{d=1}^h  \left(\lfloor\hat{A}_d \rfloor + c w_0 \right )   
   \le   C_{max}^*(\hat{I}) + h \cdot c w_0      
  \le    (1 + \delta)   C_{max}^*(I_{big})   + \tfrac{15 \delta \alpha n  w_0 } { c m} + 2 h \cdot c w_0.
$$ 
Assuming $ m  \le       \tfrac { n \cdot \alpha \cdot \epsilon}{20 c^2  \cdot h  } = \tfrac { n \cdot \alpha \cdot \delta}{c^2  \cdot h  }$, then $ h \cdot c   w_0 \le  \tfrac{\delta \alpha n \cdot w_0}{cm}$, and combining with the above inequalities, we get 
\begin{eqnarray*} 
\hat{A} 
& \le & (1 + \delta)   C_{max}^*(I_{big})   + \tfrac{17 \delta \alpha n  w_0 } { c m}    \\
& \le &  (1 + \delta) C_{max}^*(I_{round})  + \tfrac{17 \delta \alpha n  w_0 }{c m}  \hspace{1.3in }  \text { by } (\ref{round-vs-big-alpha})\\
& \le &  (1 + \delta) C_{max}^*(I_{round}) +    17 \delta C_{max}^*(I) \hspace{1.05in } \text { by } C_{max}^*(I) \ge \tfrac{\alpha n \cdot w_0}{c m} \\ 
& \le &  (1 + \delta)^2 C_{max}^*(I) +    17 \delta C_{max}^*(I) \hspace{1.3in }  \text { by } (\ref{round-vs original})  \\ 
& \le &  (1 + 20 \delta)  C_{max}^*(I)   \\
& \le  &   (  1 +   \epsilon ) C_{max}^*(I) \hspace{2.35in } \text { by } \delta =  \tfrac{\epsilon }{20 }  
\end{eqnarray*}
and 
\begin{eqnarray*} 
\hat{A} &  >  & C_{max}^*(\hat{I}) \\
&  >  & (1 - \delta) ( C_{max}^*(I_{big}) - h\cdot cw_0) \hspace{2in }  \text { by } (\ref{big-vs-estimate-1-alpha}) \\
&  > & ( 1 - \delta) ((1 - \delta)C_{max}^*(I_{round})  - \tfrac{8 \delta \alpha n w_0 }{c m} - h \cdot cw_0) \hspace{0.7in }  \text { by } (\ref{round-vs-big-alpha})\\
&  > & ( 1 - \delta) ((1 - \delta)C_{max}^*(I_{round})  - \tfrac{9 \delta \alpha n w_0 }{c m} )  \hspace{1.3in }  \text { by } h \cdot c   w_0 \le  \tfrac{\delta \alpha n \cdot w_0}{cm} \\
&  > & ( 1 - \delta) \left(     (1 - \delta)   C_{max}^*(I) - \tfrac{9 \delta \alpha n w }{c m}\right)  \hspace{1.6in }  \text { by } (\ref{round-vs original})\\
&  > & ( 1 - \delta) \left(     (1 - \delta)   C_{max}^*(I) -  9 \delta C_{max}^*(I)\right)  \hspace{1.3in }    \text { by } C_{max}^*(I) \ge \tfrac{\alpha n  w_0 }{c m}\\
& \ge &      (1 - \delta)  ( 1 - 10 \delta) C_{max}^*(I)  \\
& = &   (  1 -  \epsilon )C_{max}^*(I) \hspace{2.9in } \text { by } \delta =  \tfrac{\epsilon }{20} 
\end{eqnarray*}
\end{proof}

From Theorem~\ref{thm-sublinear-time-alg}, we can easily get the following corollaries. 
\begin{corollary}\label{cor:sublinear-time2} 
When
$m = {o}(n^{ 1/2})$, 
Randomized-Algorithm2 is a randomized $(1+\epsilon)$-approximation scheme  for $P \mid prec, dp_j \le h,  p_{[n]} \le c \cdot p_{[ (1 - \alpha) n)]} \mid C_{max}$, 
that runs in sublinear time.
\end{corollary}

\begin{corollary} For any  $\alpha =n^{-\phi}$ where   $\phi\in (0,1/3)$, if $ m  \le    \tfrac { n \cdot \alpha \cdot \epsilon}{20 c^2  \cdot h  } $, there is a randomized $(1+\epsilon)$-approximation scheme for $P \mid prec, dp_j \le h,  p_{[n]} \le c \cdot p_{[ (1 - \alpha) n)]} \mid C_{max}$
  and the algorithm runs in sublinear time  $O( \tfrac{c^4 h^2} {\alpha^3 \epsilon^6} \cdot  m^2 \log ^ 2 (\tfrac{c n}{\epsilon})  \log (\tfrac{h}{\epsilon} \log (\tfrac{c n}{\epsilon}  )))$.
\end{corollary}

Clearly, the algorithm will also work if  there is no precedence constraint, i.e. all jobs have the same depth 1. This becomes the traditional load balancing problem.
\begin{corollary} 
 For any  $\alpha =n^{-\phi}$ where   $\phi\in (0,1/3)$, if $ m  \le    \tfrac { n \cdot \alpha \cdot \epsilon}{20 c^2  \cdot h  } $,  there is a  randomized $(1+\epsilon)$-approximation scheme for $P \mid  p_{[n]} \le c \cdot p_{[ (1 - \alpha) n)]} \mid C_{max}$
 and the algorithm runs in sublinear time  $O( \tfrac{c^4} {\alpha^3 \epsilon^6} \cdot    m^2 \log ^ 2 (\tfrac{c n}{\epsilon})  \log (\tfrac{1}{\epsilon} \log (\tfrac{c n}{\epsilon}  )))$.
\end{corollary}

\subsection{The Estimated Sketch of the Schedule}
In this subsection, we will show that as  the streaming algorithms in Section~\ref{sec:stremaing-alg}, the two sublinear time algorithms in this section can compute an {\bf estimated sketch of a schedule}   $\widehat{SKS} = \{t_d: 1 \le d \le h\}$ which describes a schedule where all the jobs of depth $d$ are scheduled during the interval $[t_{d-1}, t_d)$ for $1 \le d \le h$ where $t_0 = 0$.
And we will show that  based on $\widehat{SKS}$,  the Algorithm SketchToSchedule in Section~\ref{sec:Stream-Sketch-Schedule} can, with high probability, generate a feasible schedule with the makespan at most $(1+ 2\epsilon)$ times the optimal makespan.

For the problem $P \mid prec, dp_j \le h, p_{min} \le c \cdot p_{max},  \mid C_{max}$, 
we let the estimate sketch of a schedule be $\widehat{SKS} = \{t_d: 1 \le d \le h\}$ where $t_d = \sum_{i=1}^d (\lfloor\tfrac{\hat{A}_d}{1-\delta} \rfloor + c +  3 \lfloor\tau(n,h,c)\rfloor \cdot k   \cdot c )$ for all $1 \le d \le h$.  
We have the following theorem for the estimate sketch of a schedule:

 \begin{theorem}\label{thm: sketch-alg4}
  For $P \mid prec, dp_j \le h,  p_{min} \le c \cdot p_{max}
 \mid C_{max}$, given   any $0 < \epsilon < 1$, when $m \le  \tfrac {n \epsilon}{20 h \cdot c}$,  Randomized-Algorithm1  can compute  an estimated sketch of the schedule $\widehat{SKS}$,  and based on the sketch, with probability at least $\tfrac{9}{10}$, Algorithm SketchToSchedule can generate  a feasible schedule with the makespan at most $(1+2\epsilon)$ times the optimal makespan.
 \end{theorem}
 
 \begin{proof}It is easy to see that 
Randomized-Algorithm1 can compute  the estimate sketch of the schedule $\widehat{SKS} = \{t_d: 1 \le d \le h\}$.
Now we will show that with high probability all the jobs with depth $d$ can be feasibly scheduled during the interval $[t_{d-1}, t_d)$ for $1 \le d \le h$ where $t_0 = 0$.
It suffices to prove that with high probability all the jobs from the input instance $I$ with depth $d$ can be scheduled within an interval of length $\lfloor \tfrac{\hat{A}_d}{1-\delta} \rfloor + c + \lfloor 3 \tau(n,h,c) \rfloor   \cdot k \cdot c$. Since the instance $I_{round}$ is obtained from $I$ by rounding up the processing times, all we need to prove is that the jobs from the   instance $I_{round}$ with depth $d$ can be scheduled within an interval of length $\lfloor \tfrac{\hat{A}_d}{1-\delta} \rfloor + c + \lfloor 3 \tau(n,h,c) \rfloor   \cdot k \cdot c$. 

As the proof of Lemma~\ref{lemma:rounded-big-vs-estimated}, 
we split the jobs in $I_{round}$ into two parts: 
those $(d, u)$-groups that are  in $I_{big}$, and those  $(d, u)$-groups that are not in $I_{big}$. For the jobs from the former, with probability at least $\tfrac{9}{10}$, all $(d, u)$-group  in $I_{big}$ are  also included in $\hat{I}$ and for each $(d, u)$-group of this type, we have $n_{d,u} \ge 3\tau(n, h, c)$ and $n_{d, u}\le \tfrac{\hat{e}_{d, u}}{1-\delta}$. 
Thus all these jobs at depth $d$ 
can be feasibly scheduled during the interval of length $$\floor{  \tfrac{1}{m} \sum_{u=0}^{k} (n_{d,u} \cdot rp_u)} + c  \le  \floor {\tfrac{1}{m} \sum_{u=0}^{k} \tfrac{(\hat{e}_{d,u} \cdot rp_u) } {(1- \delta)}} + c = \floor {  \tfrac{\hat{A}_d}{1-\delta} } + c .$$ 
For the jobs from the latter $(d,u)$-groups, we have $n_{d,u}<3\tau(n, h, c)$, and thus they can be feasibly scheduled on a single machine  during an interval of length $\lfloor 3 \tau(n,h,c) \rfloor \cdot k   \cdot c  $.  Therefore, combining both types of jobs, we have that with probability at least $\tfrac{9}{10}$ all the jobs with depth $d$ from $I_{round}$ can be scheduled within a time interval of length  $\lfloor \tfrac{\hat{A}_d}{1-\delta} \rfloor + c + \lfloor 3 \tau(n,h,c) \rfloor  \cdot c \cdot k$. Since the jobs  in $I_{round}$ are rounded up from those in $I$, the jobs  depth $d$ from $I$ can also  be scheduled within the same interval length.

Finally,   it is easy to see that the Algorithm SketchToSchedule generates a feasible schedule of all the jobs with depth $d$ from $I$ 
during the interval $[t_{d-1}, t_d)$. The makespan of the final schedule is at most
$$  t_h  =  \sum_{d=1}^h \left(\floor{\tfrac{\hat{A}_d}{1-\delta}} + c +   \floor {3 \tau(n,h,c)} \cdot k \cdot c \right)  
     \le \sum_{d=1}^h \left (\tfrac{\hat{A}_d}  {1-\delta}   \right)  + h \cdot c + \floor {3 \tau(n,h,c)} \cdot h \cdot k \cdot c
$$ 
 Note $3 \tau(n,h,c) \cdot h\cdot k \cdot c \le \tfrac{ 8 \delta n}{m}$, and $h \cdot c \le \tfrac{\delta n}{m}$ when  $m \le  \tfrac {n \epsilon}{20h \cdot c}$. Thus,  
\begin{eqnarray*}
    t_h   & \le &  \sum_{d=1}^h \left (\frac{\hat{A}_d}  {1-\delta}   \right)  + \tfrac{  \delta n}{m}  +  \tfrac{ 8 \delta n}{m}  \\  
     & \le &  \frac{1}{1-\delta} \sum_{d=1}^h \left ( \hat{A}_d \right)  + \tfrac{  9 \delta n}{m}  \\ 
    & \le &  \frac{1}{1-\delta} \sum_{d=1}^h \left ( \hat{A}_d \right) 
      +  9 \delta C_{max}^*(I) \hspace{1in }\text { by } C_{max}^*(I) \ge \tfrac{n}{m} \\
    & \le &  \frac{ \hat{A}}{1-\delta} +  9 \delta C_{max}^*(I)    \\
    & \le & \frac{(1+20 \delta)}{1-\delta}   C_{max}^*(I)  +    9 \delta C_{max}^*(I)\hspace{0.8in }\text{ by Theorem~\ref{thm:const-alg} } \\
        & \le & (1+ 25 \delta)C_{max}^*(I)  +    9 \delta C_{max}^*(I)\hspace{0.9in }\text{ by } \delta = \tfrac{\epsilon}{20} < \tfrac{1}{20}   \\
    & \le & (1 + 2 \epsilon) C_{max}^*(I) .   
   \end{eqnarray*}   
\end{proof}

For $P \mid prec, dp_j \le h, p_{[n]} \le c \cdot p_{[ (1 - \alpha) n)]} \mid C_{max}$, we let the estimate sketch of a schedule be $\widehat{SKS} = \{t_d: 1 \le d \le h\}$ where $t_d = \sum_{i=1}^d (\lfloor \tfrac{\hat{A}_d}{1-\delta} \rfloor + c w_0 + \lfloor 3 \tau(n,h,c) \rfloor \cdot k \cdot c w_0  + \lfloor \delta w_0 \rfloor )$ for all $1 \le d \le h$. 
For this sketch of the schedule, we can get similar conclusion. 

 \begin{theorem}\label{thm: sketch-sublinear-time-alg}
 For the problem  $P \mid  prec, dp_j \le h, p_{[n]} \le c \cdot p_{[ (1 - \alpha) n)]} \mid C_{max}$, given   any $0 < \epsilon < 1$, and $ m  \le    \tfrac { n \cdot \alpha \cdot \epsilon}{20 c^2  \cdot h  } $,  Randomized-Algorithm2 can generate  an estimated sketch of the schedule $\widehat{SKS}$, and with probability at least $\tfrac{9}{10}$, Algorithm SketchToSchedule can produce based on $\widehat{SKS}$ a feasible schedule with the makespan at most $(1+2\epsilon)$ times the optimal makespan.
 \end{theorem}
 \begin{proof}
It is easy to see that Randomized-Algorithm2 can generate the estimated sketch of a schedule $\widehat{SKS}$.
Now we will show that with high probability all the jobs with depth $d$ can be feasibly scheduled during the interval $[t_{d-1}, t_d)$ for $1 \le d \le h$ where $t_0 = 0$.
As the proof of Theorem~\ref{thm: sketch-alg4}, it suffices to prove that with high probability 
the jobs from the   instance $I_{round}$ with depth $d$ can be scheduled within an interval of length $\lfloor \tfrac{\hat{A}_d}{1-\delta} \rfloor + c w_0 + \lfloor 3 \tau(n,h,c) \rfloor  \cdot  k  \cdot c w_0 + \lfloor \delta w_0 \rfloor$. 

For this problem, the jobs in  $I_{round}$ can be split into three types: 
(1) $(d, u)$-groups that are in $I_{big}$, (2) $(d, u)$-groups corresponding to $(d,u, n_{d,u})$ where  $n_{d,u}<3\tau(n, h, c)$ for all $1 \le d \le h$ and $1 \le u \le k$ 
and   the processing times of all jobs are greater than $\tfrac{\delta w_0}{n}$, 
and (3) jobs whose processing times are no more than $\tfrac{\delta w_0}{n}$. 
We will bound the interval length needed to schedule jobs from each type. For type (1) jobs, as the proof in   Theorem~\ref{thm-sublinear-time-alg}, with probability at least $\tfrac{9}{10}$, all $(d, u)$-group  in $I_{big}$ are  also included in $\hat{I}$ and for each $(d, u)$-group of this type, we have $n_{d,u} \ge 3\tau(n, h, c)$ and $n_{d, u}\le \tfrac{\hat{e}_{d, u}}{1-\delta}$.  
Thus all these jobs at depth $d$ 
can be feasibly scheduled during the interval of length $$\floor{  \tfrac{1}{m} \sum_{u=0}^{k} (n_{d,u} \cdot rp_u)} + c w_0 \le  \floor {\tfrac{1}{m} \sum_{u=0}^{k} \tfrac{(\hat{e}_{d,u} \cdot rp_u) } {(1- \delta)}} + c w_0 
= \floor {  \tfrac{\hat{A}_d}{1-\delta} } + c w_0.$$ 
For type (2) jobs, it is easy to see that all the jobs at  depth $d$ 
 can be feasibly scheduled  during an interval of length $\lfloor 3 \tau(n,h,c) \rfloor \cdot c w_0 \cdot  k$; 
 For type (3) jobs,  since the processing times are integer,  the processing time of each job must be at most $\lfloor \tfrac{\delta w_0}{n} \rfloor$. There are at most $n$ such jobs at each depth, thus   they can be feasibly scheduled during an interval of length $ n     \lfloor \tfrac{\delta w_0}{n} \rfloor = \lfloor \delta w_0 \rfloor$. Adding all these together, with probability at least $\tfrac{9}{10}$ all jobs at depth $d$ from $I_{round}$ can be scheduled into an interval of length 
$$ \lfloor    \tfrac{\hat{A}_d}{1-\delta} \rfloor + c w_0 + \lfloor 3 \tau(n,h,c) \rfloor \cdot  k \cdot c w_0 +  \lfloor \delta w_0 \rfloor.$$
Similar as before,   
we can use Algorithm SketchToSchedule to generate a feasible schedule with the makespan  at most
 \begin{eqnarray*}
      t_h & = & \sum_{d=1}^h (\lfloor    \tfrac{\hat{A}_d}{1-\delta} \rfloor + c w_0 + \lfloor 3 \tau(n,h,c) \rfloor  \cdot  k \cdot c w_0+  \lfloor \delta w_0 \rfloor) \\
    & \le &  \left(\sum_{d=1}^h     \tfrac{\hat{A}_d}{1-\delta}   \right)  + h \cdot \left(c w_0 + 3 \tau(n,h,c)    \cdot    k \cdot c w_0 +    \delta w_0   \right) \\
          & \le &  \left(\sum_{d=1}^h     \tfrac{\hat{A}_d}{1-\delta}   \right)  +  (1+ \delta) h \cdot  c w_0 +  3 \tau(n,h,c)   \cdot h  \cdot    k  \cdot c w_0   
\end{eqnarray*}
By line {\ref{line:random2-tau}} of   Randomized-Algorithm2, $ 3 \tau(n,h,c)   \cdot h \cdot    k \cdot c w_0    \le \tfrac{8 \delta \alpha n \cdot w_0}{cm} $ which implies 
  $h \cdot  c w_0 \le \tfrac{\delta \alpha n \cdot w_0}{cm} $. With  $m  \le \frac { n \cdot \alpha \cdot \epsilon}{20 c^2  \cdot h  }$,  we get 
   $t_h
\le    \left(\sum_{d=1}^h     \frac{\hat{A}_d}{1-\delta}   \right)  + 
 (1+\delta)\tfrac{\delta \alpha n \cdot w_0}{cm}
 + \tfrac{8\delta \alpha n \cdot w_0}{cm}  $. Since $C_{max}^*(I) \ge \tfrac{\alpha n \cdot w_0}{c m}$, we have $t_h \le   \left(\sum_{d=1}^h     \tfrac{\hat{A}_d}{1-\delta}   \right)  + 
  10\delta  C_{max}^*(I)  \le        \frac{\hat{A}}{1-\delta}  
                        +   10\delta  C_{max}^*(I)$.  By Theorem~\ref{thm:const-alg}, and $\delta = \tfrac{\epsilon}{20}$, we have 
    $$t_h
      \le    \frac{(1+20 \delta)}{1-\delta}  C_{max}^*(I)  + 10 \delta C_{max}^*(I)    \le   (1 + 2 \epsilon) C_{max}^* . $$  
This completes the proof.
\end{proof}

\section{Conclusions}

In this work,  we studied the parallel machine precedence constrained scheduling problems 
  $P \mid  prec, dp_j \le h, p_{max} \le c \cdot p_{min} \mid C_{max}$ and 
  $P \mid  prec, dp_j \le h, p_{[n]} \le c \cdot p_{[ (1 - \alpha) n)]} \mid C_{max}$. 
  We focused on  two types of computing paradigms, sublinear space algorithms and sublinear time algorithms, which are inspired by the boost of multitude of data 
  in manufacturing and service industry.
  It is worth mentioning that in spite of the inapproximability result that there does not exist a polynomial time approximation algorithm with approximation ratio better than $\tfrac{4}{3}$ unless P=NP, 
    our algorithms imply that both problems admit approximation schemes if $m$ satisfies certain condition.  
Moreover, 
our algorithms for precedence constrained problems also imply the 
sublinear approximation algorithms for the popular load balancing problem where jobs are independent.  

Our work not only provides an algorithmic solutions to the studied problem under big data model, but also provide a methodological framework for designing  sublinear approximation algorithms that can be used for solving other scheduling problems. 
  In particular, besides outputting the approximate value of the optimal makespan, we introduced the concept of ``the sketch of a schedule'' to cater the need of generating a concrete schedule which approximates the optimal schedule. 
For our studied problems, it is also interesting to design sublinear approximation algorithms for other various precedence constraints and other performance criteria including total completion time, maximum tardiness, etc.

\printbibliography 

@article{dw84,
title = {Scheduling precedence graphs of bounded height},
journal = {Journal of Algorithms},
volume = {5},
number = {1},
pages = {48-59},
year = {1984},
author = {Danny Dolev and Manfred K Warmuth}
}

@InProceedings{ma00,
author="Ma, Bin",
editor="Giancarlo, Raffaele
and Sankoff, David",
title="A Polynomial Time Approximation Scheme for the Closest Substring Problem",
booktitle="Combinatorial Pattern Matching",
year="2000",
publisher="Springer Berlin Heidelberg",
address="Berlin, Heidelberg",
pages="99--107",
}

@article{graham66,
  title={Bounds for certain multiprocessing anomalies},
  author={Graham, Ronald L},
  journal={Bell system technical journal},
  volume={45},
  number={9},
  pages={1563--1581},
  year={1966},
  publisher={Wiley Online Library}
}

@book{mrp13,
author = {Motwani, Rajeev and Raghavan, Prabhakar},
title = {Randomized Algorithms},
year = {1995},
publisher = {Cambridge University Press},
address = {Cambridge, UK},
}

@article{pb18,
author = {Prot, D. and Bellenguez-Morineau, O.},
title = {A Survey on How the Structure of Precedence Constraints May Change the Complexity Class of Scheduling Problems},
year = {2018},
issue_date = {February 2018},
publisher = {Kluwer Academic Publishers},
address = {USA},
volume = {21},
number = {1},
journal = {Journal  of Scheduling},
month = {2},
pages = {3–16}
}

@article{am06,
  author    = {Isto Aho and  Erkki M{\"{a}}kinen},
  title     = {On a parallel machine scheduling problem with precedence constraints},
  journal   = {Journal of Scheduling},
  volume    = {9},
  number    = {5},
  pages     = {493--495},
  year      = {2006},
}

@article{u75,
title = {NP-complete scheduling problems},
journal = {Journal of Computer and System Sciences},
volume = {10},
number = {3},
pages = {384-393},
year = {1975},
author = {J.D. Ullman},
}

@INPROCEEDINGS{bk09,
  author={Bansal, Nikhil and Khot, Subhash},
  booktitle={2009 50th Annual IEEE Symposium on Foundations of Computer Science}, 
  title={Optimal Long Code Test with One Free Bit}, 
  year={2009},
  volume={},
  number={},
  pages={453-462},
 
}

@article{lk78,
 author = {J. K. Lenstra and A. H. G. Rinnooy Kan},
 journal = {Operations Research},
 number = {1},
 pages = {22--35},
 publisher = {INFORMS},
 title = {Complexity of Scheduling under Precedence Constraints},
 volume = {26},
 year = {1978}
}

@article{cg72,
author = {Coffman, E. G. and Graham, R. L.},
title = {Optimal Scheduling for Two-Processor Systems},
year = {1972},
volume = {1},
number = {3},
journal = {Acta Informatica},
month = {9},
pages = {200–213},
}

@article{s11,
author = {Svensson, Ola},
title = {Hardness of Precedence Constrained Scheduling on Identical Machines},
journal = {SIAM Journal on Computing},
volume = {40},
number = {5},
pages = {1258-1274},
year = {2011},
}

@inproceedings{CzumajSohler04,
author = {Czumaj, Artur and Sohler, Christian},
title = {Estimating the Weight of Metric Minimum Spanning Trees in Sublinear-Time},
year = {2004},
location = {Chicago, IL, USA},
pages = {175-183},
series = {STOC '04}
}

@article{Feige06,
author = {Feige, Uriel},
title = {On Sums of Independent Random Variables with Unbounded Variance and Estimating the Average Degree in a Graph},
journal = {SIAM Journal on Computing},
volume = {35},
number = {4},
pages = {964-984},
year = {2006},
}

@article{FlajoletMartin85,
title = {Probabilistic counting algorithms for data base applications},
journal = {Journal of Computer and System Sciences},
volume = {31},
number = {2},
pages = {182-209},
year = {1985},
author = {Philippe Flajolet and G. {Nigel Martin}},
}

@article{FuChen06,
  title={Sublinear time width-bounded separators and their application to the protein side-chain packing problem},
  author={Bin Fu and Zhixiang Chen},
  journal={Journal of Combinatorial Optimization},
  year={2006},
  volume={15},
  pages={387-407}
}

@article{GoldreichRon00,
  title={On Testing Expansion in Bounded-Degree Graphs},
  author={Oded Goldreich and Dana Ron},
  journal={Electronic Colloquium on Computational
  Complexity},
  year={2000},
  volume={7}
}

@InProceedings{GoldreichRon05,
author={Goldreich, Oded
and Ron, Dana},
editor={D{\'i}az, Josep
and Jansen, Klaus
and Rolim, Jos{\'e} D. P.
and Zwick, Uri},
title={Approximating Average Parameters of Graphs},
booktitle={Approximation, Randomization, and Combinatorial Optimization. Algorithms and Techniques},
year={2006},
publisher={Springer Berlin Heidelberg},
address={Berlin, Heidelberg},
pages={363-374},
}

@article{MunroPaterson80,
title = {Selection and sorting with limited storage},
journal = {Theoretical Computer Science},
volume = {12},
number = {3},
pages = {315-323},
year = {1980},
author = {J.I. Munro and M.S. Paterson},
}

@article{GoldreichGoldwasserRon96,
author = {Goldreich, Oded and Goldwasser, Shari and Ron, Dana},
title = {Property Testing and Its Connection to Learning and Approximation},
year = {1998},
issue_date = {July 1998},
publisher = {Association for Computing Machinery},
address = {New York, NY, USA},
volume = {45},
number = {4},
journal = {Journal of the ACM},
month = {7},
pages = {653–750},
}

@article{AlonMatiasSzegedy96,
title = {The Space Complexity of Approximating the Frequency Moments},
journal = {Journal of Computer and System Sciences},
author = {Noga Alon and Yossi Matias and Mario Szegedy},
volume = {58},
number = {1},
pages = {137-147},
year = {1999},
}

@article{ChazelleLiuMagen05,
author = {Chazelle, Bernard and Liu, Ding and Magen, Avner},
title = {Sublinear Geometric Algorithms},
journal = {SIAM Journal on Computing},
volume = {35},
number = {3},
pages = {627-646},
year = {2005},
}

@article{CzumajErgun05,
author = {Czumaj, Artur and Erg\"{u}n, Funda and Fortnow, Lance and Magen, Avner and Newman, Ilan and Rubinfeld, Ronitt and Sohler, Christian},
title = {Approximating the Weight of the Euclidean Minimum Spanning Tree in Sublinear Time},
journal = {SIAM Journal on Computing},
volume = {35},
number = {1},
pages = {91-109},
year = {2005},
}

@article{ChazelleRubinfeldTrevisan05,
author = {Chazelle, Bernard and Rubinfeld, Ronitt and Trevisan, Luca},
title = {Approximating the Minimum Spanning Tree Weight in Sublinear Time},
journal = {SIAM Journal on Computing},
volume = {34},
number = {6},
pages = {1370-1379},
year = {2005},   
}

@article{mcgregor14,
author = {McGregor, Andrew},
title = {Graph Stream Algorithms: A Survey},
year = {2014},
issue_date = {March 2014},
publisher = {Association for Computing Machinery},
address = {New York, NY, USA},
volume = {43},
number = {1},
journal = {SIGMOD Rec.},
month = {5},
pages = {9–20},
numpages = {12}
}

@article{Muthu05,
author = {Muthukrishnan, S.},
title = {Data Streams: Algorithms and Applications},
year = {2005},
issue_date = {August 2005},
publisher = {Now Publishers Inc.},
address = {Hanover, MA, USA},
volume = {1},
number = {2},
journal = {Foundations and Trends in Theoretical Computer Science},
month = {8},
pages = {117–236},
}

\end{document}